\documentclass[10pt]{IEEEtran}
\usepackage{epsfig,color,amsmath,cite}
\usepackage{amsthm}
\usepackage{amsmath}    
\IEEEoverridecommandlockouts
\usepackage{bm}
\usepackage{epstopdf}
\usepackage{amssymb}
\usepackage{url}
\usepackage{enumitem}
\usepackage{multirow}
\usepackage{hhline}
\usepackage{booktabs}
\usepackage{algorithm,algorithmic}	
\usepackage{comment}
\setlist[itemize]{leftmargin=*}

\DeclareMathOperator*{\minimize}{minimize}

\DeclareMathOperator*{\subjectto}{subject\ to}

\makeatother
\DeclareMathAlphabet\mathbfcal{OMS}{cmsy}{b}{n}

\newtheorem{theorem}{Theorem}
\newtheorem{mydef}{Definition}

\newtheorem{myrem}{Remark}

\newtheorem{myprs}{Proposition}



\usepackage{stackengine}

\newcommand{\mat}[1]{\boldsymbol{#1}}

\newcommand{\bmat}[1]{\begin{bmatrix} #1 \end{bmatrix}}

\providecommand{\mA}{\ensuremath{\mat{A}}}

\providecommand{\mC}{\ensuremath{\mat{C}}}

\providecommand{\mI}{\ensuremath{\mat{I}}}

\providecommand{\mL}{\ensuremath{\mat{L}}}

\providecommand{\mO}{\ensuremath{\mat{O}}}
\providecommand{\mP}{\ensuremath{\mat{P}}}

\providecommand{\mY}{\ensuremath{\mat{Y}}}
\providecommand{\mZ}{\ensuremath{\mat{Z}}}

\newcommand{\m}{\boldsymbol}
\allowdisplaybreaks[4]
\usepackage[colorlinks = true,
linkcolor = blue,
urlcolor  = blue,
citecolor = blue,
anchorcolor = blue]{hyperref}

\newcommand{\mc}[1]{\mathcal{#1}}

\usepackage[framemethod=TikZ]{mdframed}
\mdfdefinestyle{MyFrame}{%
	linecolor=black,
outerlinewidth=0.5pt,
	roundcorner=1.0pt,
	innerrightmargin=3pt,
	innerleftmargin=3pt,}
	
\usepackage[noabbrev]{cleveref}

\usepackage{mathtools}

\DeclarePairedDelimiter\abs{\lvert}{\rvert}%
\DeclarePairedDelimiter\norm{\lVert}{\rVert}%

\makeatletter
\let\oldabs\abs
\def\abs{\@ifstar{\oldabs}{\oldabs*}}
\let\oldnorm\norm
\def\norm{\@ifstar{\oldnorm}{\oldnorm*}}
\makeatother

\usepackage[english]{babel}
\usepackage[utf8]{inputenc}
\usepackage[super]{nth}

\usepackage{graphicx}
\usepackage{float}
\usepackage[caption = false]{subfig}

\title{\color{black} A Control-Theoretic Approach for Scalable and Robust Traffic Density Estimation using Convex Optimization}

\author {Sebastian~A.~Nugroho$^{*}$,~\IEEEmembership{Student Member,~IEEE,} Ahmad~F.~Taha$^{*}$,~\IEEEmembership{Member,~IEEE,} and \\Christian G. Claudel$^{\dagger}$,~\IEEEmembership{Member,~IEEE}.

	\thanks{
		*Department of Electrical and Computer Engineering, The University of Texas at San Antonio, 1 UTSA Circle, San Antonio, TX 78249.
		$^\dagger$Department of Civil, Architectural, and Environmental Engineering, The University of Texas at Austin, 301 E. Dean Keeton St. Stop C1700, Austin, TX 78712.
		Emails: sebastian.nugroho@my.utsa.edu, ahmad.taha@utsa.edu, christian.claudel@utexas.edu. This work was partially supported by the National Science Foundation under Grants 1636154, 1728629, 1917164, and 1917056.}
}

\begin{document}
	
\setlength{\abovedisplayskip}{3.5pt}
\setlength{\belowdisplayskip}{3.5pt}
\setlength{\abovedisplayshortskip}{3.2pt}
\setlength{\belowdisplayshortskip}{3.2pt}
	
\newdimen\origiwspc%
\newdimen\origiwstr%
\origiwspc=\fontdimen2\font
\origiwstr=\fontdimen3\font

\fontdimen2\font=0.63ex

\maketitle

\begin{abstract}
Monitoring and control of traffic networks represent alternative, inexpensive strategies to minimize traffic congestion. As the number of traffic sensors is naturally constrained by budgetary requirements, real-time estimation of traffic flow in road segments that are not equipped with sensors is of significant importance---thereby providing situational awareness and guiding real-time feedback control strategies. To that end, firstly we build a generalized traffic flow model for stretched  highways with arbitrary number of ramp flows based on the \textit{Lighthill Whitham Richards} (LWR) flow model. Secondly, we characterize the function set corresponding to the nonlinearities present in the LWR model, and use this characterization to design real-time and robust state estimators (SE) for stretched highway segments. Specifically, we show that the nonlinearities from the derived models are locally Lipschitz continuous by providing the analytical Lipschitz constants. Thirdly, the analytical derivation is then incorporated through a robust SE method given a limited number of traffic sensors, under the impact of process and measurement disturbances and unknown inputs. The estimator is based on deriving a convex semidefinite optimization problem. Finally, numerical tests are given showcasing the applicability, scalability, and robustness of the proposed estimator for large systems under high magnitude disturbances, parametric uncertainty, and unknown inputs.
\end{abstract}


\section{Introduction and Paper Contributions}

\IEEEPARstart{T}{raffic}  congeston is a growing concern in most urban areas of the world. In the US alone, congestion caused a burden of more than \$300B in 2016~\cite{inrx}. At the network level, congestion occurs when the demand exceeds the transportation network's capacity. Several strategies exist to mitigate the impact of traffic congestion, including infrastructure modifications (creation of additional lanes), capacity improvements through the use of automation (cooperative cruise control or autonomous vehicles), and {traffic network control}. Among all strategies, network control strategies are the most cost-effective and easy to implement. These include for example dynamic speed limits~\cite{hegyi2010dynamic}, ramp metering~\cite{papageorgiou2002freeway}, or dynamic toll pricing~\cite{gardner2013development}.

While traffic control strategies can be effective, they require the network operator to \textit{estimate} a real-time traffic state with the highest possible accuracy---since it is financially infeasible to install sensors at each road/highway segment. The practice of running a control scheme with incorrect traffic estimates could result in a worsening of overall traffic congestion. For this reason, a significant number of control systems are open-loop, including pre-timed traffic signals~\cite{dion2004comparison}, pre-times tolling strategies, and pre-timed ramp metering.    

To perform real-time monitoring and control of traffic networks, physics-based models are needed. Among traffic flow models, we can broadly distinguish two classes: \textit{macroscopic models}~\cite{lebacque2005first}, which compute the evolution of the vehicular density, and \textit{microscopic models}~\cite{li2017vehicle}, which model the trajectory of each vehicle. Macroscopic flow models are suitable to traffic state estimation since they scale well to large networks: the computational time required to simulate traffic is independent of the number of vehicles modeled, as opposed to microscopic models. Moreover, in macroscopic model, the number of states is fixed since it depends on the highway size.

It is known that traffic density is one of major indicators that is useful for determining traffic conditions \cite{khan2017real}. To that end, here we focus on performing state estimation of traffic density on stretched highway by considering the classical macroscopic  \textit{Lighthill Whitham Richards} (LWR) flow model~\cite{Lighthill1995b,Richards1956}, which is a first order hyperbolic conservation law. This nonlinear dynamic model has been extensively used in traffic modeling for various purposes, e.g., in estimation and control applications~\cite{work2008ensemble,yuan2012real} and ramp metering~\cite{Jacquet2005,Agarwal2015}. The LWR model is indeed robust and easy to calibrate, since it only depends on a small number of well-known traffic parameters, encoded as a flow-density relationship known as the \textit{fundamental diagram}. In this paper, motivated by earlier work~\cite{agarwal2016dynamic,Contreras2016,zu2018real}, we specifically utilize the \textit{Greenshield's model}~\cite{greenshields1935study} to represent the fundamental diagram, which is a concave and parabolic flow-density relationship.

The \textit{objective} of this paper is to characterize the function sets corresponding to the nonlinearities present in the LWR model, and use this characterization to design real-time state estimators for stretched highway segments having arbitrary input/output ramp flows. Specifically, we pursue a control-theoretic approach to address the state estimation problem of highways equipped with limited number of sensors; the related literature is succinctly discussed next.

{\color{black}
Many approaches have been proposed to address traffic state estimation problems. In general, methods for traffic state estimation can be categorized into model-driven and data-driven. In model-driven traffic state estimation, statistical state estimators such as Particle Filter \cite{polson2015bayesian,xia2017assimilating,Pascale2013}, Kalman Filter \cite{Liberis2016,Timotheou2015}, {Extended Kalman Filter} (EKF) \cite{tampere2007extended,wang2003motorway,Yuan2012,Yuan2014}, {Unscented Kalman Filter} (UKF) \cite{NGODUY2008599,mihaylova2006unscented}, and {Ensemble Kalman Filter} (EnKF) \cite{work2008ensemble,xia2017assimilating,SEO2015391,Gundlegard2015} are among of the most extensively used methods---see \cite[Tables 1 and 2]{seo2017traffic} for a list of state estimators used in the recent literature. To mention a few, traffic density estimation has been studied based on a switching-mode scheme of \textit{cell transmission model} (CTM) \cite{munoz2003traffic}. The same problem is revisited in \cite{alvarez2004adaptive} by developing adaptive nonlinear observers on a continuous-time traffic model, as opposed to the discrete-time CTM.
{\color{black} The traffic density estimation using Kalman Filter based on the measurements of average traffic speed and flow for mixed traffic is conducted in \cite{Liberis2016}, in which METANET traffic flow model \cite{papageorgiou1989macroscopic,papageorgiou2010traffic} is used for validation purpose.} 
{\color{black} In \cite{BEKIARISLIBERIS2017convec}, traffic state estimation using EKF based on data-driven model obtained from connected vehicles are proposed, in which the location of traffic sensors are determined using the notion of structural observability.}
The authors in \cite {tampere2007extended} implement the EKF to perform traffic density estimation based on linearized CTM. To mitigate the impact of time-delays, CTM-based decentralized observer for traffic density estimation is developed in \cite{Guo2017}. 

Recently, in addition to the traffic density estimation problem, the sensor placement problem of highway segments is studied in  \cite{Contreras2016} where the \textit{linearized} Greenshield's model is considered.
Albeit it offers simplicity in contrast with nonlinear dynamic models,
the linearized models are only representative of the dynamics when the traffic density lies in the vicinity of that point. Moreover, the study is conducted to the extent of observability of the linearized dynamic models. 
This paper aims to investigate the robust traffic density estimation problem using control-theoretic approach by considering the nonlinear nature of traffic dynamics based on Greenshield's model, while incorporating worst-case disturbance scenarios thereby yielding a robust state estimation routine.
}

The estimation of traffic density can be performed by implementing a suitable dynamic state estimation methods such as robust Kalman filters or observers. As the paper's contribution is focused on observer designs for nonlinear systems, we succinctly discuss relevant research studies pertaining to this approach.
There are indeed numerous observer design methods available in the literature. The analysis on the stability of observers for Lipschitz nonlinear systems\footnote{The nonlinear system $\dot{\m x}=\m f(\m x, \m u)$ is globally Lipschitz if there exists a constant $\gamma \geq 0$ such that
	$\norm{\m f(\m x, \m u)-\m f(\hat{\m x}, \m u)}_2 \leq \gamma \norm{\m x -  \hat{\m x} }_2$
	for all $\m x, \hat{\m x} \in \mathbb{R}^{m}$. The constant $\gamma$ essentially characterizes this nonlinearity.} is performed in \cite{Rajamani1998}. The authors in \cite{Rajamani2010} derive linear matrix inequality (LMI) conditions to synthesize observers for Lipschitz nonlinear systems; a similar result is also proposed in \cite{alessandri2004design}. Albeit these lead to relatively simple procedures, they are not designed for systems with unknown inputs, disturbances, and measurement noise which are always present in practical situations. To that end, a robust $\mathcal{H}_{\infty}$ observer for Lipschitz nonlinear systems is proposed in \cite{abbaszadeh2006robust}. Recently, authors in \cite{chakrabarty2017state} use the concept of $\mathcal{L}_{\infty}$  stability, reported in~\cite{pancake2002analysis}, to design an observer for systems with \textit{incremental quadratic nonlinearity} which is more generalized form of Lipschitz nonlinearity. An earlier version of this work appeared in~\cite{nugroho2019traffic} where we \textit{(i)} consider traffic density modeling for mostly the uncongested mode and \textit{(ii)} utilize a dynamic observer that is not designed to deal with uncertainty in process and measurement models.

In light of the aforementioned literature, the paper's contributions and organization are summarized next. 
\begin{itemize}
	\item From a transportation network modeling perspective, we formulate the traffic dynamic model of stretched highway consisting of arbitrary number and location of input and output ramp flows based on Greenshield's fundamental diagram for \textit{congested} and \textit{uncongested} modes. The modeling presumes the knowledge of congested and uncongested \textit{modes} (or \textit{cases}) on highways, which can be done through data analytics or fault detection techniques. 
	Given this, the formulated dynamic model is then represented in state-space form making it amenable to a plethora of control-theoretic approaches. This contribution is presented in Section \ref{sec:model}.
	\item From a nonlinear traffic model perspective, we prove that the nonlinearity in the Greeshield dynamic traffic model follows the \textit{locally Lipschitz continuous} function set. We also provide methods to compute the corresponding Lipschitz constant for an arbitrary highway configuration for congested and uncongested modes. This contribution is provided in Section \ref{sec:nonlinearity}. These two aforementioned contributions pave the way for two applications in transportation systems: \textit{(A)} Performing robust state estimation of traffic density by utilizing nonlinear observers for Lipschitz systems, considering uncertainty in process and measurement models. This is akin to designing $\mathcal{H}_{\infty}$ controllers to perform state feedback control.
	\textit{(B)} Building localized observer-based control strategies for ramp metering and state-feedback control. 
	\item From a control and estimation-theoretic perspective, 
we consider the aforementioned $\mathcal{L}_{\infty}$ stability concept to design an observer for traffic density estimation for systems under disturbances, unknown inputs, and sensor faults. Instead of using the {incremental quadratic nonlinearity} classification of the nonlinearity from \cite{chakrabarty2017state}, we propose a different condition which is simpler as it is designed specifically for Lipschitz nonlinear systems---a property which we prove in this paper for the traffic dynamics. Section \ref{sec:robust_obs} presents this contribution. The design of this robust observer is performed using scalable {semidefinite programming} (SDP) methods and shown to perform well---and sometimes outperforming classical Kalman-filter based estimation techniques---even under significant disturbances and parametric uncertainty; this is discussed in the numerical tests in Section \ref{sec:numerical}. 
\end{itemize}

{\color{black} It must be noted that the presented, continuous-time robust traffic estimation framework developed in this paper only considers Greenshield's fundamental diagram in continuous-time. Nonetheless, the proposed method can be extended for traffic density estimation based on CTM using other fundamental diagrams in discrete-time.} The next section presents the notation used in this paper.

\section{Notations and Preliminaries}
Italicized, boldface upper and lower case characters represent matrices and column vectors: $a$ is a scalar, $\m a$ is a vector, and $\m A$ is a matrix. Matrix $\m I$ denotes the identity square matrix, whereas $\mO$ denotes a zero matrix of appropriate dimensions. The notations $\mathbb{R}$, $\mathbb{R}_+$, and $\mathbb{R}_{++}$ denote the set of real numbers, non-negative, and positive real numbers. The notations $\mathbb{R}^n$ and $\mathbb{R}^{p\times q}$ denote row vectors with $n$ elements and matrices with size $p$-by-$q$ with elements in $\mathbb{R}$, whereas $\mathbb{S}^{m}_{+}$ and $\mathbb{S}^{m}_{++}$ denote the set of positive semi-definite and positive definite matrices. For any vector $\m x \in \mathbb{R}^{n}$, $\Vert\m x\Vert_2$ denotes the Euclidean norm of of $\m x$, defined as $\Vert \m x\Vert_2 = \sqrt{\m x^{\top}\m x} $ , where $\m x^{\top}$ is the transpose of $\m x$. 
For set $\mathcal{X}$, the notation $\abs{\mathcal{X}}$ denotes the cardinality of $\mathcal{X}$. For simplicity, the notation `$*$' denotes terms induced by symmetry in symmetric block matrices. Tab.~\ref{tab:notation} provides nomenclature utilized in the ensuing sections. 
{\color{black} In what follows, we present the formal definitions of Lipschitz continuity and the $\mathcal{L}_{\infty}$ norm.

\begin{mydef}[Lipschitz Continuity]\label{def:Lipschitz_cont}
	Let $\m f : \mathbb{R}^{m} \rightarrow \mathbb{R}^{n}$. Then, $\m f$ is Lipschitz continuous in $\mathbfcal{B} \subseteq \mathbb{R}^{m}$ if there exists a constant $\gamma \in \mathbb{R}_{+}$ such that
	\begin{align}
	\norm{\m f(\m x)-\m f(\hat{\m x})}_2 \leq \gamma \norm{\m x -  \hat{\m x} }_2, \label{lipschitz}
	\end{align}
	for all $\m x, \hat{\m x} \in \mathbfcal{B}$.
\end{mydef}

\begin{mydef}\label{L_inf_space_norm}
	The $\mathcal{L}_{\infty}$ space is defined as
	\begin{align*}
	\mathcal{L}_{\infty} \triangleq \{\m v:[0,\infty)\rightarrow \mathbb{R}^n\,|\,\norm{\m v(t)}_{\mathcal{L}_{\infty}}< \infty\},
	\end{align*}
	in which the $\mathcal{L}_{\infty}$ norm, denoted by $\norm{\cdot}_{\mathcal{L}_{\infty}}$, is defined as
	\begin{align*}
	\norm{\m v(t)}_{\mathcal{L}_{\infty}} \triangleq \sup_{t \geq 0} \,\norm{\m v(t)}_2,
	\end{align*}
	for a continuous function $\m v:[0,\infty)\rightarrow \mathbb{R}^n$.   
\end{mydef}
}

\begin{table}[t!]
	\footnotesize	\renewcommand{\arraystretch}{1.5}
	\caption{Paper nomenclature: parameter, variable, and set definitions.}
	\label{tab:notation}
	\centering
	\begin{tabular}{|l|l|}
		\hline
		\textbf{Notation} & \textbf{Description}\\
		\hline
		\hline
		\hspace{-0.1cm}$\mathbfcal{E}$ & \hspace{-0.1cm}the set of highway segments on the stretched highway \\
		\hspace{-0.1cm} & \hspace{-0.1cm}$\mathbfcal{E} = \{ 1,2,\hdots,N \}$ , $N \triangleq \abs{\mathbfcal{E}}$ \\
		\hline
		\hspace{-0.1cm}$\mathbfcal{E}_I$ & \hspace{-0.1cm}the set of highway segments with on-ramps \\
	\hspace{-0.1cm}	&  \hspace{-0.1cm}$\mathbfcal{E}_I = \{ 1,2,\hdots,N_I \}$ , $N_I \triangleq \abs{\mathbfcal{E}_I}$ \\
		\hline
		\hspace{-0.1cm}$\mathbfcal{E}_O$ & \hspace{-0.1cm}the set of highway segments with off-ramps \\
		\hspace{-0.1cm}& \hspace{-0.1cm}$\mathbfcal{E}_O = \{ 1,2,\hdots,N_O \}$, $N_O \triangleq \abs{\mathbfcal{E}_O}$ \\ 
		\hline
		\hspace{-0.1cm}$\hat{\mathbfcal{E}}$ & \hspace{-0.1cm}the set of on-ramps, $\hat{\mathbfcal{E}} = \{ 1,2,\hdots,N_I \}$ , $N_I = |\hat{\mathbfcal{E}}|$\\
		\hline
		\hspace{-0.1cm}$\check{\mathbfcal{E}}$ & \hspace{-0.1cm}the set of off-ramps, $\check{\mathbfcal{E}} = \{ 1,2,\hdots,N_O \}$ , $N_O = |\check{\mathbfcal{E}}| $\\
		\hline
		\hspace{-0.1cm}$\rho_i (t) \triangleq \rho_i$ & \hspace{-0.1cm}traffic density in segment $i \in \mathbfcal{E}$ (vehicles/m) \\
		\hline
		\hspace{-0.1cm}$q_i(t)\triangleq q_i$ & \hspace{-0.1cm}traffic flow in segment $i \in \mathbfcal{E}$ (vehicles/s)  \\
		\hline
		\hspace{-0.1cm}$v_i(t) \triangleq  v_i$ & \hspace{-0.1cm}traffic speed in segment $i \in \mathbfcal{E}$ (m/s)  \\
		\hline
		\hspace{-0.1cm}$\hat{\rho}_i(t) \triangleq \hat{\rho}_i$ & \hspace{-0.1cm}traffic density on on-ramp $i\in \hat{\mathbfcal{E}}$ \\
		\hline
		\hspace{-0.1cm}$\check{\rho}_i (t)\triangleq \check{\rho}_i$ & \hspace{-0.1cm}traffic density on off-ramp $i\in \check{\mathbfcal{E}}$ \\
		\hline
		\hspace{-0.1cm}$\hat{q}_i(t) \triangleq \hat{q}_i$ & \hspace{-0.1cm}traffic flow on the other end of on-ramp $i\in \hat{\mathbfcal{E}}$ \\
		\hline
		\hspace{-0.1cm}$\check{q}_i(t) \triangleq \check{q}_i $ & \hspace{-0.1cm}traffic flow on the other end of off-ramp $i\in \check{\mathbfcal{E}}$ \\
		\hline
		\hspace{-0.1cm}$f_{\mathrm{in}}(t)\triangleq f_{\mathrm{in}}$ & \hspace{-0.1cm}upstream flow entering the stretched highway \\
		\hline
		\hspace{-0.1cm}$f_{\mathrm{out}}(t)\triangleq f_{\mathrm{out}}$ & \hspace{-0.1cm}downstream flow exiting the stretched highway \\
		\hline
		\hspace{-0.1cm}$\hat{f}_i (t) \triangleq  \hat{f}_i$ & \hspace{-0.1cm}upstream flow entering the on-ramp $i\in \hat{\mathbfcal{E}}$ \\
		\hline
		\hspace{-0.1cm}$\check{f}_i(t)\triangleq  \check{f}_i$ & \hspace{-0.1cm}downstream flow exiting the off-ramp $i\in \check{\mathbfcal{E}}$ \\
		\hline
		\hspace{-0.1cm}$v_f$ & \hspace{-0.1cm}free-flow speed (m/s)  \\
		\hline
		\hspace{-0.1cm}$\rho_m$ & \hspace{-0.1cm}maximum density (vehicles/m)  \\
		\hline
		\hspace{-0.1cm}$\rho_c$ & \hspace{-0.1cm}critical density (vehicles/m)  \\
		\hline
		\hspace{-0.1cm}$q_m$ & \hspace{-0.1cm}maximum flow (vehicles/s)   \\
		\hline
		\hspace{-0.1cm}$\alpha(i)$ & \hspace{-0.1cm}exit ratio for off-ramp $i\in \check{\mathbfcal{E}}$, where $\alpha(i)\in [0,1]$ \\ 
		\hline 
		\hspace{-0.1cm}$\delta$ & \hspace{-0.1cm}constant equal to $\frac{v_f}{l \rho_m}$ \\
		\hline
	\end{tabular}
\end{table}

\section{Dynamic Modeling of Highway Traffic with Ramp Flows}\label{sec:model}
In this study, we consider a macroscopic traffic model referred to as the \textit{Lighthill-Whitman-Richards} (LWR) Model \cite{Lighthill1995a,Lighthill1995b,Richards1956}. This model is a nonlinear first-order hyperbolic PDE based on the vehicle conservation principle. This principle describes the evolution of traffic density on a highway segment, given the knowledge of initial conditions and boundary conditions, and is expressed by the following PDE \cite{Kachroo2008}
\begin{align}
\frac{\partial}{\partial t}\rho (x,t) + \frac{\partial}{\partial x}q(x,t) = 0, \label{eq:LWR_model}
\end{align}
where $\rho$ and $q$ are functions of position $x$ and time $t$. The relation between $\rho$ and $q$ is given as \cite[Section 2.4.1]{Kachroo2008}
\begin{align}
q(x,t) = v(x,t)  \rho(x,t), \label{eq:density_vs_flow}
\end{align}  
where $v$ is the average traffic speed. In practice, the traffic speed on a highway segment depends on its current traffic density. One of the widely used model that describes this relation is the \textit{Greenshield fundamental diagram}. This model assumes a linear relationship between traffic speed and traffic density \cite[Section 2.4.3]{Kachroo2008}, that is, for known free-flow speed $v_f$ and maximum density $\rho_m$, traffic speed is calculated as
\begin{align}
v(x,t) = v_f \left( 1-\frac{\rho(x,t)}{\rho_m}\right), \label{eq:greenshield_model}
\end{align}
as depicted in Fig. \ref{fig1}. Critical density $\rho_c$ is regarded as the density for which the flow is maximal, which for a Greenshield model corresponds to one half the maximum density, i.e., $\rho_c = \tfrac{1}{2}\rho_m$. In This model, the relation between traffic flow and traffic density is illustrated in Fig. \ref{fig1}, where the maximum flow $q_m$ is achieved when the traffic density is equal to critical density.
\begin{figure}[t]	
\centering	\includegraphics[scale=0.240]{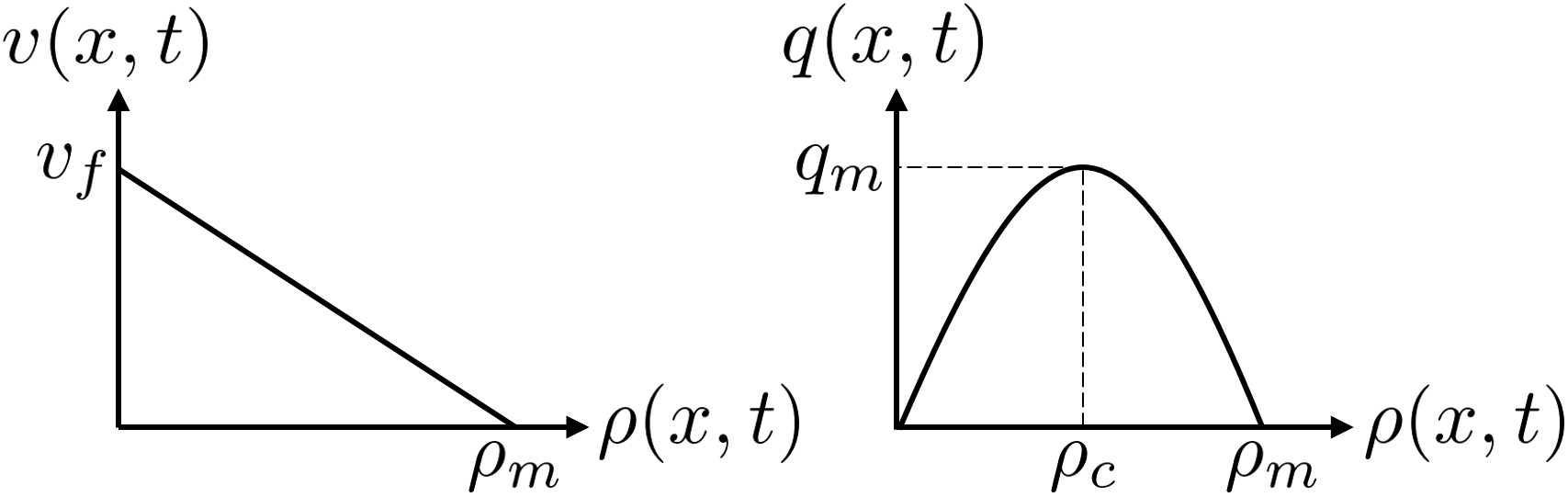}
		\vspace{-0.2cm}
	\noindent \caption{Greenshield's fundamental diagram: (left) Traffic speed versus traffic density. (right) Traffic flow versus traffic density.}
	\label{fig1}
\end{figure}

In order to enable state estimation for traffic density at highway segments without traffic sensor installation, the first-order PDE \eqref{eq:LWR_model} is then discretized in space using Godunov's scheme \cite{Lebacque1995}. The resulting model can be approximated with the following ODE 
\begin{align}
\dot{\rho} (x,t)  
  &\approx  \frac{q (x,t)-q (x+l,t)}{l}, \label{eq:discrete_LWR_model}
\end{align}
where $l$ is the length of the highway segment. The evolution of the traffic density on each highway segment is influenced by the upstream and downstream densities from its neighboring highway segments. There are two regions where the steady-state flow (or \textit{equilibirum point}) of the stretched highway can lie: congested and uncongested \cite{Contreras2016}.  {\color{black} A highway segment is said to be congested (jammed) if the density $\rho(t)$ satisfies $\rho_c < \rho(t) \leq \rho_m$, and otherwise uncongested (free-flow) if $\rho(t)$ satisfies $0 \leq \rho(t) \leq \rho_c$. This implies that, for a stretched highway divided into $N$ segments, there are $2^{N}$ possible modes, in which each segment can be either congested or uncongested. 

Following \cite{Contreras2016}, here we consider two cases that likely prevail on a stretched highway, in which segments inside a stretched highway \textit{section} (which is made of a number of segments) are all either uncongested or congested.} 
Other studies, such as~\cite{munoz2003traffic}, have considered more than two modes and then perform traffic density estimation. The mode identification is based on density measurements at cell boundaries. Other works have considered the mode identification~\cite{LEMARCHAND2012648} which we consider as given in the paper. That is, the proposed methods in this paper are not considered with the problem of modes detection and identification; we are rather concerned with the classification of the nonlinear dynamics  and scalable, robust state estimation methods.
This approach is simple since we do not consider all possible modes, thus avoiding a more complex switching-model. Furthermore, and by using the nonlinear model, one does not need to perform linearization which can be more practical in the situation when the equilibrium point is unknown.

\begin{figure}[t]	\hspace{0.0cm}\centerline{\includegraphics[keepaspectratio=true,scale=0.274]{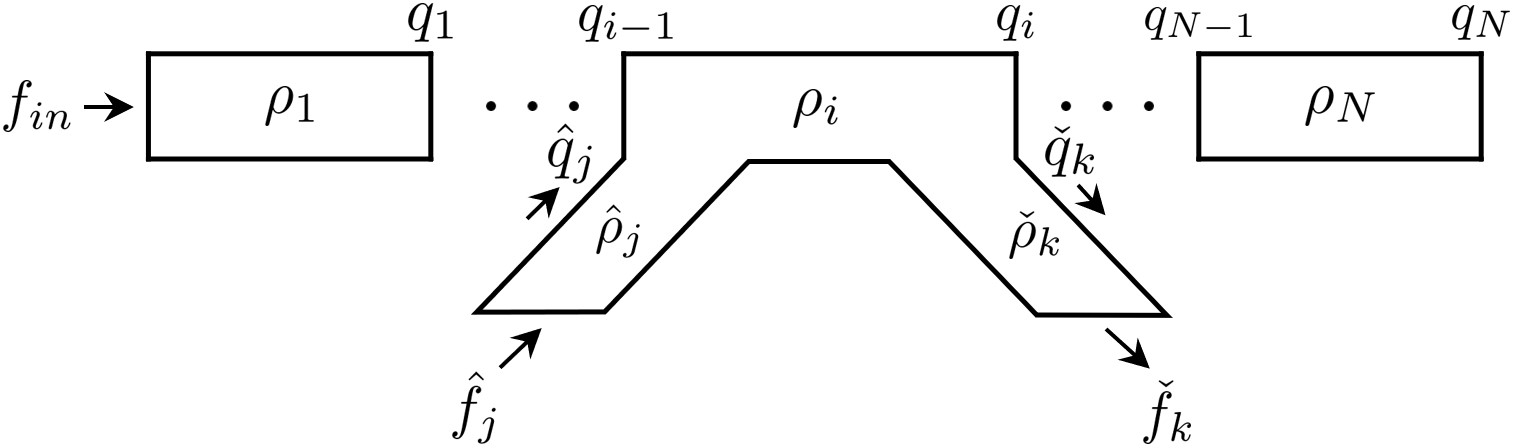}}
	\noindent \caption{The diagram of the highway for the uncongested case.}
	\label{fig2}
\end{figure}

\subsection{The Uncongested Case}

For each highway segment, \eqref{eq:discrete_LWR_model} can be generalized as follow
\begin{align}
\dot{\rho} (x,t)  
  &=  \frac{\sum q (x_1,t)-\sum q (x_2,t)}{l}, \label{eq:discrete_LWR_model_generalized}
\end{align} 
where $x_1$ and $x_2$ represent the location of the boundaries such that $x_2-x_1 = l$. In  the above equation, $\sum q (x_1,t)$ and $\sum q (x_2,t)$ denote the total inflow and total outflow associated with that highway segment. To build the dynamic model, the stretched highway is divided into $N$ segments of equal length $l$ such that the rate of change of the traffic density on each segment can be modeled by \eqref{eq:discrete_LWR_model_generalized}. {\color{black}Fig. \ref{fig2} gives an illustration on how the model is built for the uncongested case, in which we assume that $0 \leq \rho_i \leq \rho_c$ for $i \in {\mathbfcal{E}}$ and $0 \leq \rho_i \leq \rho_m$ for $i \notin {\mathbfcal{E}}$. 
For the sake of simplicity, all on- and off-ramps are assumed to have the same properties as those on the stretched highway segments, such as $l$, $v_f$, $\rho_c$ and $\rho_m$.

In addition to the above, it is also assumed that any highway segment can have at most one on-ramp and/or one off-ramp, in which the first and last highway segments are not connected to any on-ramp nor off-ramp.} If a highway segment $i$ is connected to both on- and off-ramp, then $i \in {\mathbfcal{E}}_I \cap {\mathbfcal{E}}_O$ with $N_{IO} = \vert{\mathbfcal{E}}_I \cap {\mathbfcal{E}}_O\vert$ denotes the number of highway segments connected to both on- and off-ramps. Moreover, we require that the upstream flow on the first highway segment $f_{\mathrm{in}}$, upstream flow on each on-ramp $\hat{f}_{i}$ for $i\in \hat{\mathbfcal{E}}$, and downstream flow on each off-ramp $\check{f}_{i}$ for $i\in \check{\mathbfcal{E}}$ are all  known, which in a real situation, can be obtained from conventional traffic detectors \cite{Liberis2016}. {\color{black} The exit ratio for all off-ramps are also assumed to be known and fixed.}
Based on these assumptions, by conveniently defining constant $\delta \triangleq \frac{v_f}{l \rho_m}$, and combining equations \eqref{eq:density_vs_flow} and \eqref{eq:greenshield_model} with \eqref{eq:discrete_LWR_model_generalized}, the equations describing the evolution of traffic densities can be classified into several categories, each of which are specified as follows
\begin{subequations}\label{eq:uncongested_dynamic_1}
\begin{enumerate}[label=$\alph*$)]
\item $i\in\mathbfcal{E}\setminus {\mathbfcal{E}}_I \cup {\mathbfcal{E}}_O$, $i = 1$
\begin{align}
\hspace{-0.4cm}\dot{\rho}_i &= \frac{f_{\mathrm{in}}-q_i}{l} = \frac{f_{\mathrm{in}}}{l}-\frac{v_f}{l}\rho_i+\delta \rho_i^2 \label{eq:uncongested_dynamic_1a}
\end{align}
\item $i\in\mathbfcal{E}\setminus {\mathbfcal{E}}_I \cup {\mathbfcal{E}}_O$, $i \neq 1$
\begin{align}
\hspace{-0.4cm}\dot{\rho}_i &= \frac{q_{i-1}-q_i}{l} 
= \frac{v_f}{l}\left(\rho_{i-1}-\rho_i\right)-\delta \left(\rho_{i-1}^2-\rho_i^2 \right)\label{eq:uncongested_dynamic_1b}
\end{align}
\item $i\in\mathbfcal{E}_I\setminus{\mathbfcal{E}}_I \cap {\mathbfcal{E}}_O$, $j\in\hat{\mathbfcal{E}}$
\begin{align}
\hspace{-0.6cm}\dot{\rho}_i &= \frac{q_{i-1}+\hat{q}_j-q_{i}}{l} = \frac{v_f}{l}(\rho_{i-1}-\rho_i+\hat{\rho}_j) \nonumber \\ & \quad -\delta \left(\rho_{i-1}^2-\rho_i^2+\hat{\rho}_j^2\right)\label{eq:uncongested_dynamic_1c}
\end{align}
\item $i\in\mathbfcal{E}_O\setminus{\mathbfcal{E}}_I \cap {\mathbfcal{E}}_O$, $j\in\check{\mathbfcal{E}}$
\begin{align}
\hspace{-0.6cm}\dot{\rho}_i &= \frac{q_{i-1}-\alpha(j)\check{q}_j-q_{i}}{l} = \frac{v_f}{l}(\rho_{i-1}-\rho_i -\alpha(j)\check{\rho}_j) \nonumber \\ & \quad-\delta \left(\rho_{i-1}^2-\rho_i^2-\alpha(j)\check{\rho}_j^2\right)\label{eq:uncongested_dynamic_1d}
\end{align}
\item $i\in{\mathbfcal{E}}_I \cap {\mathbfcal{E}}_O$, $j\in\hat{\mathbfcal{E}}$, $k\in\check{\mathbfcal{E}}$
\begin{align}
\hspace{-0.6cm}\dot{\rho}_i &= \frac{q_{i-1}+\hat{q}_j-\alpha(k)\check{q}_k-q_{i}}{l}\nonumber \\ 
&= \frac{v_f}{l}(\rho_{i-1}-\rho_i+\hat{\rho}_j-\alpha(k)\check{\rho}_k) \nonumber \\ & \quad -\delta \left(\rho_{i-1}^2-\rho_i^2+\hat{\rho}_j^2-\alpha(k)\check{\rho}_k^2\right) \label{eq:uncongested_dynamic_1e}
\end{align}
\item $i\in\hat{\mathbfcal{E}}$
\begin{align}
\hspace{-0.4cm}\dot{\hat{\rho}}_i &= \frac{\hat{f}_{i}-\hat{q}_i}{l} = \frac{\hat{f}_{i}}{l}-\frac{v_f}{l}\hat{\rho}_i+\delta \hat{\rho}_i^2 \label{eq:uncongested_dynamic_1f}
\end{align}
\item $i\in\check{\mathbfcal{E}}$
\begin{align}
\hspace{-0.4cm}\dot{\hat{\rho}}_i &= \frac{\alpha(i)\check{q}_i-\check{f}_i}{l} = -\frac{\check{f}_{i}}{l}+\alpha(i)\frac{v_f}{l}\check{\rho}_i-\alpha(i)\delta \check{\rho}_i^2. \label{eq:uncongested_dynamic_1g}
\end{align}
\end{enumerate}
\end{subequations}
We construct the state vector $\m x \triangleq \bmat{\rho_i & \cdots & \hat{\rho}_j & \cdots & \check{\rho}_k}^{\top}$ in an increasing order for all $i\in{\mathbfcal{E}}$, $j\in\hat{\mathbfcal{E}}$, and $k\in\check{\mathbfcal{E}}$ such that $\m x$ {\color{black} is of dimension $n$} where $n = N+N_I+N_O$. By doing so, Eq.~\eqref{eq:uncongested_dynamic_1} can be written in nonlinear state-space format 
\begin{mdframed}[style=MyFrame]
\begin{align}
\dot{\m x} (t)&= \bmat{\m A_1 & \m A_2 \\ \m O & \m A_3} \m x (t)+ 	\m f(\m x) + \m {B_{\mathrm{u}}} \m u(t), \label{eq:state_space_uncongested_ramps}
\end{align}
\end{mdframed}
with $\m A_1$, $\m A_2$, $\m A_3$, $\m f(\cdot)$, $\m {B_{\mathrm{u}}}$, and $\m u$ are specified in Tab.~\ref{tab:dynamics_parameter_uncongested} of Appendix~\ref{apdx:table_congested}. {\color{black}
	\begin{myrem}\label{rem:set_X_U}
		The uncongested case (or region) considers that $x_i\in [0,\rho_c]$ for all $i\in\mathbfcal{E}$ and $x_i\in [0,\rho_m]$ otherwise. For convenience, we define the set $\mathbfcal{X}_{\m{\mathrm{u}}}\triangleq [0,\rho_c]^N\times \hat{\mathbfcal{X}}\times\check{\mathbfcal{X}}$ where  $\hat{\mathbfcal{X}} \triangleq [0,\rho_m]^{N_I}$ and $\check{\mathbfcal{X}} \triangleq [0,\rho_m]^{N_O}$ such that $\m x\in \mathbfcal{X}_{\m{\mathrm{u}}}$. Therefore, any traffic condition in the uncongested case represented by \eqref{eq:state_space_uncongested_ramps} and Tab.~\ref{tab:dynamics_parameter_uncongested} is assumed to have at least one equilibrium point inside $\mathbfcal{X}_{\m{\mathrm{u}}}$.
	\end{myrem}}

\subsection{The Congested Case}
{\color{black} In this section, the modeling for the traffic model considering congested zones is presented, in which all highway segments are assumed to be congested, i.e., $\rho_c < \rho_i \leq \rho_m$ for $i \in {\mathbfcal{E}}$ and $0 \leq \rho_i \leq \rho_m$ for $i \notin {\mathbfcal{E}}$.} We consider the same assumptions as in the uncongested case except that the downstream flow on the last highway segment $f_{\mathrm{out}}$ is known (instead of $f_{\mathrm{in}}$). To that end, the evolution of traffic densities in the congested case are formulated as follows
\begin{subequations}\label{eq:uncongested_dynamic_2}
\begin{enumerate}[label=$\alph*$)]
\item $i\in\mathbfcal{E}\setminus {\mathbfcal{E}}_I \cup {\mathbfcal{E}}_O$, $i = N$
\begin{align}
\hspace{-0.4cm}\dot{\rho}_i &= \frac{q_i-f_{\mathrm{out}}}{l} = \frac{v_f}{l}\rho_i-\delta \rho_i^2-\frac{f_{\mathrm{out}}}{l} \label{eq:uncongested_dynamic_2a}
\end{align}
\item $i\in\mathbfcal{E}\setminus {\mathbfcal{E}}_I \cup {\mathbfcal{E}}_O$, $i \neq N$
\begin{align}
\hspace{-0.4cm}\dot{\rho}_i &= \frac{q_{i}-q_{i+1}}{l} 
= \frac{v_f}{l}\left(\rho_{i}-\rho_{i+1}\right)-\delta \left(\rho_{i}^2-\rho_{i+1}^2 \right)\label{eq:uncongested_dynamic_2b}
\end{align}
\item $i\in\mathbfcal{E}_I\setminus{\mathbfcal{E}}_I \cap {\mathbfcal{E}}_O$, $j\in\hat{\mathbfcal{E}}$
\begin{align}
\hspace{-0.6cm}\dot{\rho}_i &= \frac{q_{i}+\hat{q}_j-q_{i+1}}{l} = \frac{v_f}{l}(\rho_{i}-\rho_{i+1} +\hat{\rho}_j)\nonumber \\ & \quad -\delta \left(\rho_{i}^2-\rho_{i+1}^2+\hat{\rho}_j^2\right)\label{eq:uncongested_dynamic_2c}
\end{align}
\item $i\in\mathbfcal{E}_O\setminus{\mathbfcal{E}}_I \cap {\mathbfcal{E}}_O$, $j\in\check{\mathbfcal{E}}$
\begin{align}
\hspace{-0.6cm}\dot{\rho}_i &= \frac{q_{i}-\alpha(j)\check{q}_j-q_{i+1}}{l} = \frac{v_f}{l}(\rho_{i}-\rho_{i+1}-\alpha(j)\check{\rho}_j)\nonumber \\ & \quad -\delta \left(\rho_{i}^2-\rho_{i+1}^2-\alpha(j)\check{\rho}_j^2\right)\label{eq:uncongested_dynamic_2d}
\end{align}
\item $i\in{\mathbfcal{E}}_I \cap {\mathbfcal{E}}_O$, $j\in\hat{\mathbfcal{E}}$, $k\in\check{\mathbfcal{E}}$
\begin{align}
\hspace{-0.6cm}\dot{\rho}_i &= \frac{q_{i}+\hat{q}_j-\alpha(k)\check{q}_k-q_{i+1}}{l}\nonumber \\ 
&= \frac{v_f}{l}(\rho_{i}-\rho_{i+1}+\hat{\rho}_j-\alpha(k)\check{\rho}_k) \nonumber \\ & \quad -\delta \left(\rho_{i}^2-\rho_{i+1}^2+\hat{\rho}_j^2-\alpha(k)\check{\rho}_k^2\right) \label{eq:uncongested_dynamic_2e}
\end{align}
\item $i\in\hat{\mathbfcal{E}}$
\begin{align}
\hspace{-0.4cm}\dot{\hat{\rho}}_i &= \frac{\hat{f}_{i}-\hat{q}_i}{l} = \frac{\hat{f}_{i}}{l}-\frac{v_f}{l}\hat{\rho}_i+\delta \hat{\rho}_i^2 \label{eq:uncongested_dynamic_2f}
\end{align}
\item $i\in\check{\mathbfcal{E}}$
\begin{align}
\hspace{-0.4cm}\dot{\hat{\rho}}_i &= \frac{\alpha(i)\check{q}_i-\check{f}_i}{l} = -\frac{\check{f}_{i}}{l}+\alpha(i)\frac{v_f}{l}\check{\rho}_i -\alpha(i)\delta \check{\rho}_i^2. \label{eq:uncongested_dynamic_2g}
\end{align}
\end{enumerate}
\end{subequations}
Likewise, from \eqref{eq:uncongested_dynamic_2}, the traffic dynamics for the congested case can be written in a state-space form of \eqref{eq:state_space_uncongested_ramps}, where $\m A_1$, $\m f(\cdot)$, $\m {B_{\mathrm{u}}}$, and $\m u$ are detailed in Tab.~\ref{tab:dynamics_parameter_congested} and  matrices $\m A_2$ and $\m A_3$ are given in Tab.~\ref{tab:dynamics_parameter_uncongested} of Appendix~\ref{apdx:table_congested}. Given the dynamic models for the congested and uncontested cases, the next section presents a characterization of the function set of the nonlinearity $\m f(\cdot)$.

{\color{black}	\begin{myrem}\label{rem:set_X_C}
Since in the congested case it is assumed that $x_i\in (\rho_c,\rho_m]$ for all $i\in\mathbfcal{E}$ and $x_i\in [0,\rho_m]$ otherwise, we define $\mathbfcal{X}_{\m{\mathrm{c}}}\triangleq [\rho_c,\rho_m]^N\times \hat{\mathbfcal{X}}\times\check{\mathbfcal{X}}$ where  $\hat{\mathbfcal{X}}$ and $\check{\mathbfcal{X}}$ are defined in Remark \ref{rem:set_X_U} such that $\m x\in \mathbfcal{X}_{\m{\mathrm{c}}}$. Therefore, any traffic condition in the congested case represented by \eqref{eq:state_space_uncongested_ramps} and Tables  \ref{tab:dynamics_parameter_uncongested} and \ref{tab:dynamics_parameter_congested} is assumed to have at least one equilibrium point inside $\mathbfcal{X}_{\m{\mathrm{c}}}$.
\end{myrem}}
\section{Characterization of the Nonlinear Functions}\label{sec:nonlinearity}
This section discusses the investigation of the function set class for $\m f(\cdot)$. As there are several function sets that a multi-variable function can belong to (e.g., locally/globally Lipschitz, one-sided Lipschitz, quadratically bounded, etc.). From a control-theoretic perspective, this is important as it allows the design of asymptotic state-feedback controllers and state observers---the latter application is discussed in Section~\ref{sec:robust_obs}.
   In the case of traffic dynamics, the fact that the traffic density is differentiable and bounded in its domain implies that the nonlinear term $\m f(\cdot)$ is differentiable Lipschitz continuous---at least locally in a bounded region of the state-space. {\color{black}The definition of Lipschitz continuity is described in \eqref{lipschitz}.} 
   
Although the smallest $\gamma$ satisfying \eqref{lipschitz} is more desirable, finding such constant can be cumbersome. For practical purpose, however, finding any $\gamma$ that satisfies \eqref{lipschitz} while still useful for designing controllers and observers is sufficient. With that in mind, we present analytical methods to determine Lipschitz constants for nonlinear function $\m f(\cdot)$ for both uncongested and congested cases.

{\color{black}
\begin{myprs}\label{proposition1}
In the uncongested case, the nonlinear function $\m f:\mathbb{R}^n\rightarrow\mathbb{R}^n$ governing the traffic dynamics \eqref{eq:state_space_uncongested_ramps} and specified in Tab.~\ref{tab:dynamics_parameter_uncongested} is locally Lipschitz in $\mathbfcal{X}_{\m{\mathrm{u}}}$ with 
\begin{mdframed}[style=MyFrame]
\vspace{-0.20cm}
\footnotesize \begin{align}
\gamma_u &= \frac{v_f}{l}\left(\vphantom{\sum_{{i\in\mathbfcal{E}\setminus {\mathbfcal{E}}_I \cup {\mathbfcal{E}}_O}}\frac{v_f^2}{l^2}}2N+2N_I-1+(6+4\sqrt{2})(N_I-N_O+N_{IO})\right.\nonumber\\
& \quad \left. +   \sum_{i\in\check{\mathbfcal{E}}\setminus \hat{\mathbfcal{E}} \cap \check{\mathbfcal{E}}}\left(4\sqrt{2}\alpha(i)+4\alpha^2(i)\right)+\sum_{i\in{\check{\mathbfcal{E}}}}4\alpha^2(i)\right.\nonumber\\
& \quad \left. + \sum_{i\in\hat{\mathbfcal{E}} \cap \check{\mathbfcal{E}}}\left((8+4\sqrt{2})\alpha(i)+4\alpha^2(i)\right)\vphantom{\sum_{{i\in\mathbfcal{E}\setminus {\mathbfcal{E}}_I \cup {\mathbfcal{E}}_O}}\frac{v_f^2}{l^2}}\right)^{\!1/2}. \label{eq:lipschitz_const_uncongested}
\end{align}
\end{mdframed}
\end{myprs}}
\vspace{-0.2cm}
\begin{proof}
	{\color{black}
	See Appendix \ref{apdx:proposition1_proof}.}  
\end{proof}

{\color{black}
	\begin{myprs}\label{proposition2}
		In the congested case, the nonlinear function $\m f:\mathbb{R}^n\rightarrow\mathbb{R}^n$ governing the traffic dynamics \eqref{eq:state_space_uncongested_ramps} and specified in Tab.~\ref{tab:dynamics_parameter_uncongested} is locally Lipschitz in $\mathbfcal{X}_{\m{\mathrm{c}}}$ with 
		\begin{mdframed}[style=MyFrame]
			\vspace{-0.20cm}
			\footnotesize \begin{align}
			\gamma_c &= \frac{2v_f}{l}\left(\vphantom{\sum_{{i\in\mathbfcal{E}\setminus {\mathbfcal{E}}_I \cup {\mathbfcal{E}}_O}}\frac{v_f^2}{l^2}}2N+3N_I-1+\sum_{i\in\check{\mathbfcal{E}}\setminus \hat{\mathbfcal{E}} \cap \check{\mathbfcal{E}}}\left(2\sqrt{2}\alpha(i)+\alpha^2(i)\right)\right.\nonumber\\
			& \quad \left. + \sum_{i\in\hat{\mathbfcal{E}} \cap \check{\mathbfcal{E}}}\left(4\alpha(i)+\alpha^2(i)\right)+\sum_{i\in{\check{\mathbfcal{E}}}}\alpha^2(i)\vphantom{\sum_{{i\in\mathbfcal{E}\setminus {\mathbfcal{E}}_I \cup {\mathbfcal{E}}_O}}\frac{v_f^2}{l^2}}\right)^{\!1/2}. \label{eq:lipschitz_const_congested}
			\end{align}
		\end{mdframed}
\end{myprs}}
\vspace{-0.2cm}
\begin{proof}
	{\color{black}
		See Appendix \ref{apdx:proposition2_proof}.}  
\end{proof}

\begin{myrem}
	The analytically derived Lipschitz constants $\gamma_u$ and $\gamma_c$ depend on the traffic network parameters, number of in- and out-flow ramps, and how many highway segments are labeled as congested and uncongested. Therefore, and depending on the classification of traffic modes in highways, the constants $\gamma_u$ and $\gamma_c$ will change. 
	\end{myrem}
It is important to notice that the results given in Theorem \ref{proposition1} and Corollary \ref{proposition2} are useful in the sense that, not only are they showing the nonlinearities to be Lipschitz, but more importantly, they also provide practical Lipschitz constants that can be actually used in Lipschitz-based state estimator and observer designs. Furthermore, these results make it possible to implement many methods in control theory which can potentially solve problems arising in traffic networks that can be cast into estimation and control problems. 
To that end, the next section proposes a robust state estimation method for the nonlinear dynamics~\eqref{eq:state_space_uncongested_ramps} considering limited number of measurements and unknown disturbances.  

\vspace{-0.1cm}
\section{{\color{black}Robust State Estimation Via $\mc{L}_{\infty}$ Observer}} \label{sec:robust_obs}
 As shown in previous sections, the traffic density of highways with multiple ramp flows can be modeled by a set of ODE, which is further represented by a nonlinear state-space equation. Moreover, it is also shown that the nonlinearity satisfies the local Lipschitz condition and the Lipschitz constant is derived for any arbitrary highway configuration. This enables us to design a certain type of observer to perform robust traffic density estimation, which is the focus of this section.

\subsection{Traffic Modeling under Disturbances and $\mc{L}_{\infty}$ Stability}

{\color{black} Here, we present the perturbed traffic dynamics. Specifically, we consider that the perturbation or uncertainty is due to unknown inputs, measurement inaccuracies, process disturbances, measurement noise, and sensor faults.} These can all be succinctly represented by vector $\m w(t)$, an unknown quantity. Given these disturbances, the nonlinear perturbed dynamics can be expressed as
\begin{subequations}\label{eq:state_space_general}
\begin{align}
\dot{\m x}(t) &= \mA \m x(t) + 	\m f(\m x) + \m {B_{\mathrm{u}}}\m u (t)+ \m {B_{\mathrm{w}}} \m w(t) \label{eq:state_space_general_a}\\
\m y(t) &= \mC \m x(t)+ \m {D_{\mathrm{w}}} \m w(t). \label{eq:state_space_general_b}
\end{align}
\end{subequations}
In the above model, \eqref{eq:state_space_general_a} represents \eqref{eq:state_space_uncongested_ramps} with unknown inputs $ \m {B_{\mathrm{w}}} \m w$ and \eqref{eq:state_space_general_b} is the linear measurement model with measurement noise $\m {D_{\mathrm{w}}} \m w$, where $\m y\in\mathbb{R}^p$ is the measurement vector and $\m C\in\mathbb{R}^{p\times n}$ is a matrix representing the configuration and location of the sensors. The disturbance vector $\m w\in\mathbb{R}^q$ is assumed to be bounded, with the corresponding matrices $\m {B_{\mathrm{w}}}$ and $\m {D_{\mathrm{w}}}$ are of appropriate dimensions. 
{\color{black}
Note that vector $\m w$ lumps all unknown inputs into a single vector.
For example, if $\m v_1$ represents unknown inputs and $\m v_2$ represents measurement noise, with corresponding matrices $\m {V_1}$ and $\m {V_2}$, such that the system dynamics are expressed as
\begin{subequations}\label{eq:state_space_general_rem}
	\begin{align}
	\dot{\m x}(t) &= \mA \m x(t) + 	\m f(\m x) + \m {B_{\mathrm{u}}}\m u (t)+ \m {V_1}  \m v_1(t)\\
	\m y(t) &= \mC \m x(t)+ \m {V_2} \m v_2(t),
	\end{align}
\end{subequations}
then defining $\m w = \bmat{\m v_1^{\top} & \m v_2^{\top}}^{\top}$ along with $\m {B_{\mathrm{w}}} = \bmat{\m {V_1} & \m O}$ and $\m {D_{\mathrm{w}}} = \bmat{\m O & \m {V_2}}$ of appropriate dimensions allows \eqref{eq:state_space_general_rem} to be expressed in form of \eqref{eq:state_space_general}. 
}
In this study, we are interested in the case when many highway segments do not have traffic sensor installations, i.e., $p < n$. Thus, the objective of the observer is to estimate the traffic density for the entire highway segments.

To design the observer, let $\hat{\m x}(t)$ be the observer's state (or estimation) vector and $\hat{\m y}(t)$ be the observer's measurement vector. The proposed observer dynamics follow a similar form to the classic Luenberger observer, and are given as
\begin{subequations} \label{eq:nonlinear_observer_dynamics}
\begin{align}	
\dot{\hat{\m x}}(t)&= \mA \hat{\m x}(t)+\m f(\hat{\m x}) + \m {B_{\mathrm{u}}}\m u(t) + \mL(\m y(t)-\hat{\m y}(t)) \\
\hat{\m y}(t) &= \mC \hat{\m x}(t),
\end{align}
\end{subequations}
where $\mL(\m y-\hat{\m y})$ is the Luenberger-type correction term with $\mL\in\mathbb{R}^{n\times p}$. {\color{black} In order to ensure the existence of such observer, it is assumed that the traffic sensors have been placed in such a way that they yield the pair $(\m A,\m C)$ detectable.} By defining the estimation error as $\m e(t)  \triangleq \m x (t)- \hat{\m x}(t)$, the error dynamics can be computed as 
\begin{align}
\hspace{-0.3cm}\dot{\m e} (t) = \left(\mA-\mL\mC\right)\m e(t) +\Delta \m f (t)+ \left(\m {B_{\mathrm{w}}}-\mL\m {D_{\mathrm{w}}}\right)\m w(t), \label{eq:est_error_dynamics}
\end{align}
where $\Delta \m f(t)\triangleq \m f(\m x)-\m f(\hat{\m x})$. Since the traffic dynamic model is determined by the choice of states' operational range, which can be either congested or uncongested, then it is helpful to have the following definition.
{\color{black}
\begin{mydef}\label{def:set_x}
	The set ${\mathbfcal{X}}\subset \mathbb{R}_+^n$ is defined as
	\begin{align}\label{eq:def_set_x}
	{\mathbfcal{X}} = \begin{cases}
	\mathbfcal{X}_{\m{\mathrm{u}}}, & \text{if the highway is uncongested} \\
	\mathbfcal{X}_{\m{\mathrm{c}}}, & \text{if the highway is congested.} \\
	\end{cases}		
	\end{align}
\end{mydef}}
{\color{black}
By using Definition \ref{def:set_x}, then we can simply invoke $\m x,\,\hat{\m x}\in\mathbfcal{X}$ regardless on the condition of the highway. The notion of $\mathcal{L}_{\infty}$ stability with performance level $\mu$ is introduced below. }
{\color{black}
\begin{mydef}\label{L_inf_stability}
	{\color{black}
Let $\m z \in \mathbb{R}^z$ be a performance output constructed as $\m z = \mZ \m e$ for a user-defined performance matrix $\mZ\in\mathbb{R}^{z\times n}$. Then, the nonlinear dynamics \eqref{eq:est_error_dynamics} is said to be $\mathcal{L}_{\infty}$ stable in $\mathbfcal{X}$ with performance level $\mu$ if the following conditions hold.}
\begin{enumerate}
\item The undisturbed system is uniformly asymptotically stable around the origin.
\item For any bounded disturbance $\m w \in \mathcal{L}_{\infty}$ and zero initial error $\m e_0 = 0$, we have $\norm{\m z}_2 \leq \mu \norm{\m w}_{\mathcal{L}_{\infty}}$.
\item There exists a function $\beta:\mathbb{R}^n\times \mathbb{R}_+\rightarrow \mathbb{R}_+$ such that, for any initial error $\m e_0$ and any bounded disturbance $\m w \in \mathcal{L}_{\infty}$, we have $\norm{\m z}_2 \leq\beta\left(\m e_0,\norm{\m w}_{\mathcal{L}_{\infty}}\right)$.
\item For any initial error $\m e_0$ and any bounded disturbance $\m w \in \mathcal{L}_{\infty}$, we have $\lim_{t\to\infty}\sup\,\norm{\m z}_2 \leq\mu\norm{\m w}_{\mathcal{L}_{\infty}}$.
\end{enumerate} 
\end{mydef}}

{\color{black} 
	The definition of $\mathcal{L}_{\infty}$ stability for error dynamics \eqref{eq:est_error_dynamics} with performance level $\mu$ described in Definition \ref{L_inf_stability} can be interpreted as follows, assuming that conditions 1-4 are satisfied. First, in the case when disturbance is not present, the estimation error will asymptotically converge towards zero. Note that this property is standard in many observer designs. 
	Second, in the presence of disturbance given that the initial error is equal to zero, the norm of performance vector $\m z(t)$ for any $t \geq t_0$ is guaranteed to be no more than a scalar multiple of the worst case disturbance, that is, $\norm{\m z}_2\leq \mu\norm{\m w}_{\mathcal{L}_{\infty}}$. Third, the norm of $\m z(t)$ is always upper bounded by a function of initial condition $\m e_0$ and worst case disturbance $\norm{\m w}_{\mathcal{L}_{\infty}}$. Fourth, if the initial error is nonzero, the norm of performance vector $\m z(t)$ will evolve such that it will not exceed the value of $\mu\norm{\m w}_{\mathcal{L}_{\infty}}$.  
	Note that whenever 
	$\mu = 0$, $\m z (t)$ (therefore, the estimation error) is irrelevant to the disturbance $\m w(t)$. On the other hand, large $\mu$ implies that 
	small change in $\m w(t)$ will greatly affect $\m z (t)$. In that regard, we should always try to make performance index $\mu$ as small as possible.
	In what follows we shift to the design of numerical procedure that, if solved successfully, ensures the estimation error dynamics  \eqref{eq:est_error_dynamics} to be $\mathcal{L}_{\infty}$ stable.
	 in the sense of Definition \ref{L_inf_stability}.}

\subsection{$\mc{L}_{\infty}$ Observer Design}
{\color{black}
In this section we present a sufficient condition to synthesize the $\mathcal{L}_{\infty}$ observer for Lipschitz nonlinear systems---described in the following theorem.} 
\begin{theorem}\label{l_inf_theorem}
Consider the nonlinear system with unknown input and measurement noise \eqref{eq:state_space_general} and observer \eqref{eq:nonlinear_observer_dynamics} where $\m x,\hat{\m x}\in\mathbfcal{X}$, $\m w \in \mathcal{L}_{\infty}$, and the nonlinear function $\m f :\mathbb{R}^{n}\rightarrow\mathbb{R}^{n}$ is locally Lipschitz in $\mathbfcal{X}$ with Lipschitz constant $\gamma$. If there exist $\mP\in\mathbb{S}^n_{++}$, $\mY\in\mathbb{R}^{n\times p}$, $\epsilon,\mu_0,\mu_1,\mu_2\in \mathbb{R}_{+}$, and $\alpha \in \mathbb{R}_{++}$ so that the following optimization problem is solved
\begin{mdframed}[style=MyFrame]
\begin{subequations}\label{eq:l_inf_theorem}
\begin{align}
&\m{\mathcal{L}_{\infty}}-\m{\mathrm{Observer}} \nonumber\\
 &\minimize_{\m P, \m Y, \epsilon, \alpha, \mu_{0,1,2}} \quad \mu_0\mu_1 + \mu_2 \label{eq:l_inf_theorem_0}\\
&\subjectto \nonumber \\
&\bmat{ \mA^{\top}\mP + \mP\mA -\mC^{\top}\mY^{\top} \\ -\mY\mC+\alpha\mP+\epsilon\gamma^2\mI&*&*\\
\mP & -\epsilon\mI&*\\
\m {B_{\mathrm{w}}}^{\top}\mP-\m {D_{\mathrm{w}}}^{\top}\mY^{\top}&\mO&-\alpha\mu_0\mI} \preceq 0 \label{eq:l_inf_theorem_1}\\
&\bmat{-\mP & * & * \\
\mO & -\mu_2\mI & *\\
\mZ & \mO & -\mu_1\mI}\preceq 0, \label{eq:l_inf_theorem_2}
\end{align}
\end{subequations}
\end{mdframed}
\vspace{-0.1cm}
then the error dynamics \eqref{eq:est_error_dynamics} is $\mathcal{L}_{\infty}$ stable with performance level $\mu = \sqrt{\mu_0\mu_1+\mu_2}$ for performance output given as $\m z = \mZ \m e$. In this case, the observer gain is given as $\mL = \mP^{-1}\mY$.
\end{theorem}

\begin{proof}
	{\color{black}
		See Appendix \ref{apdx:thrm2_proof}.}  
\end{proof}

Realize that the $\mathcal{L}_{\infty}\text{-}\mathrm{Observer}$ problem  is nonconvex due to bilinear terms appearing in  \eqref{eq:l_inf_theorem_0} and \eqref{eq:l_inf_theorem_1}. Specifically, the problem is nonconvex in terms of variables $\alpha,\,\mu_0,\,\mu_1$, and $\mP$. To render  the $\mathcal{L}_{\infty}\text{-}\mathrm{Observer}$ as a convex optimization problem, one can set the values of $\alpha$ and either $\mu_0$ or $\mu_1$ a priori and solve a semidefinite program (SDP). An alternative to this approach is utilizing a successive convex approximation of the bilinear terms in~\eqref{eq:l_inf_theorem}; the authors' recent work includes examples on how this can be applied~\cite{Taha2018arXiv}.

\begin{myrem}
In the proposed observer design, we consider that the Lipschitz constants and state-space matrices are fixed. If a change to the classification of the two traffic modes takes place, the observer design problem needs to be solved again for updated observer gain matrix $\m L$. This requires a scalable SDP solver, which is discussed in the end of the next section.
\end{myrem}

\vspace{-0.4cm}
\section{\color{black}Numerical Tests under Different Scenarios}\label{sec:numerical}
The objective of this section is three-fold. First, demonstrating that the formulated SDP formulation graciously scale with the number of highway segments through a scalable SDP solver. Second, showcasing that the derived Lipschitz constant is not conservative and that the SDP for $\mc{L}_{\infty}$ problem can in fact yield feasible solutions. Third, demonstrating the applicability of the proposed state estimation method under high magnitude disturbances in the process dynamics and the measurement model. All simulations are performed using MATLAB R2019a running on a 64-bit Windows 10 with 3.4GHz Intel\textsuperscript{R} Core\textsuperscript{TM} i7-6700 CPU and 16 GB of RAM with YALMIP \cite{Lofberg2004} as the interface to solve all convex problems.{\color{black} Throughout this section, in order to make the numerical test results more intuitive for the reader, we have changed the units for traffic density $\m x(t)$, error $\m e(t)$, and error norm $\norm {\m e(t)}_2$ from vehicles/m to vehicles/km and multiply the norm of performance output $\m z(t)$ and infinity norm of disturbance $\m w(t)$ by $10^3$.}

{\color{black}
	\begin{figure}[t]
	\vspace{-0.2cm}
	\centering 
	\subfloat[]{\includegraphics[keepaspectratio=true,scale=0.4]{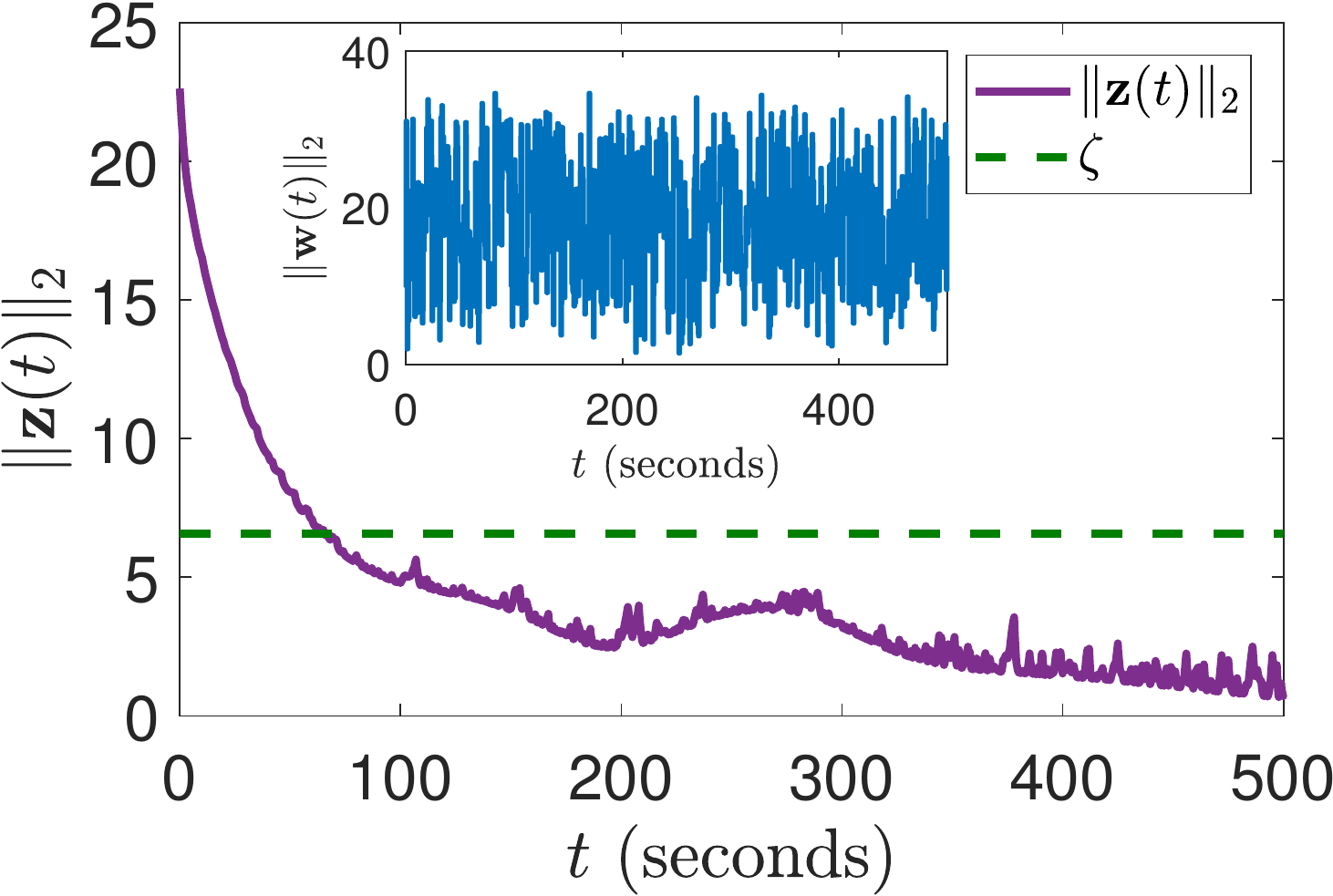}}{}\vspace{-0.3cm}
	\subfloat[]{\includegraphics[keepaspectratio=true,scale=0.4]{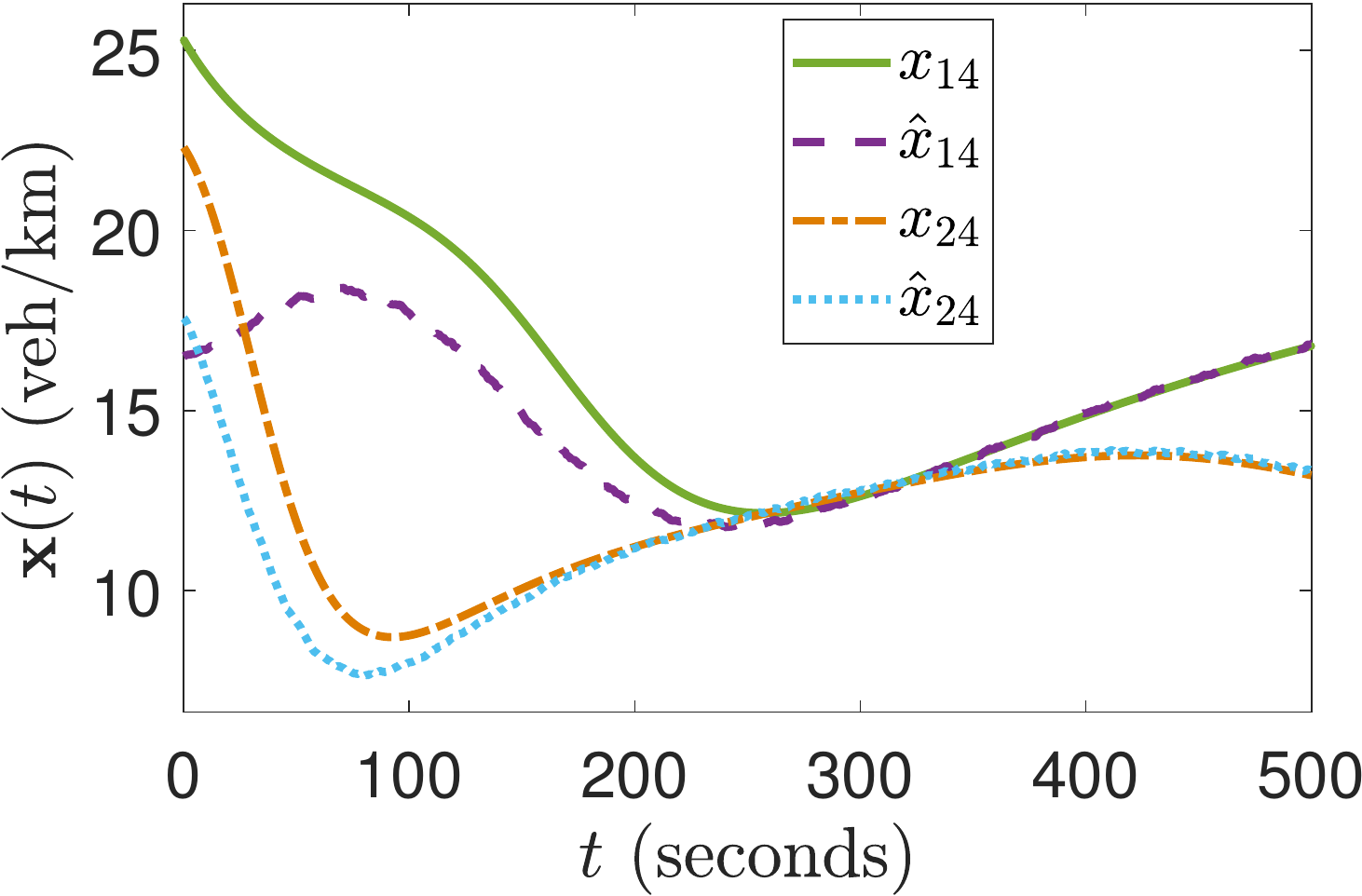}}{}\hspace{-0.0cm}\vspace{-0.1cm}
	\caption{(a) Comparison between the norm of performance output $||{\m z}(t)||_2$ and disturbance where $\zeta = \mu||{\m w}(t)||_{\mathcal{L}_{\infty}}$ for uncongested (free-flow) case and (b) the trajectories of two unmeasured states.}
	\label{fig:uncongested_uncertainty_1}\vspace{-0.4cm}
\end{figure}

\begin{figure}[t]
	\vspace{0.1cm}
	\centering 
	\subfloat[]{\includegraphics[keepaspectratio=true,scale=0.4]{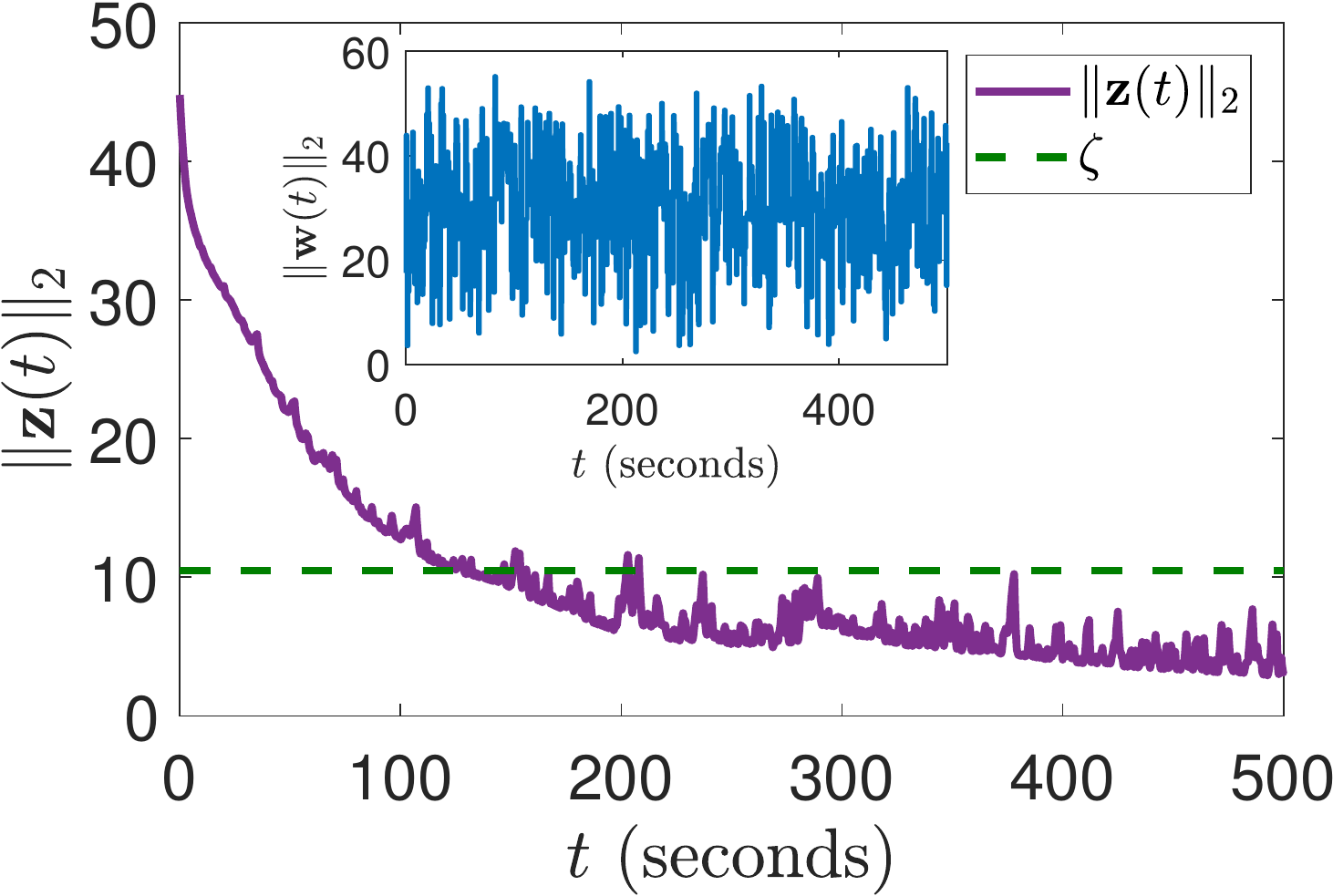}}{}\vspace{-0.3cm}
	\subfloat[]{\includegraphics[keepaspectratio=true,scale=0.4]{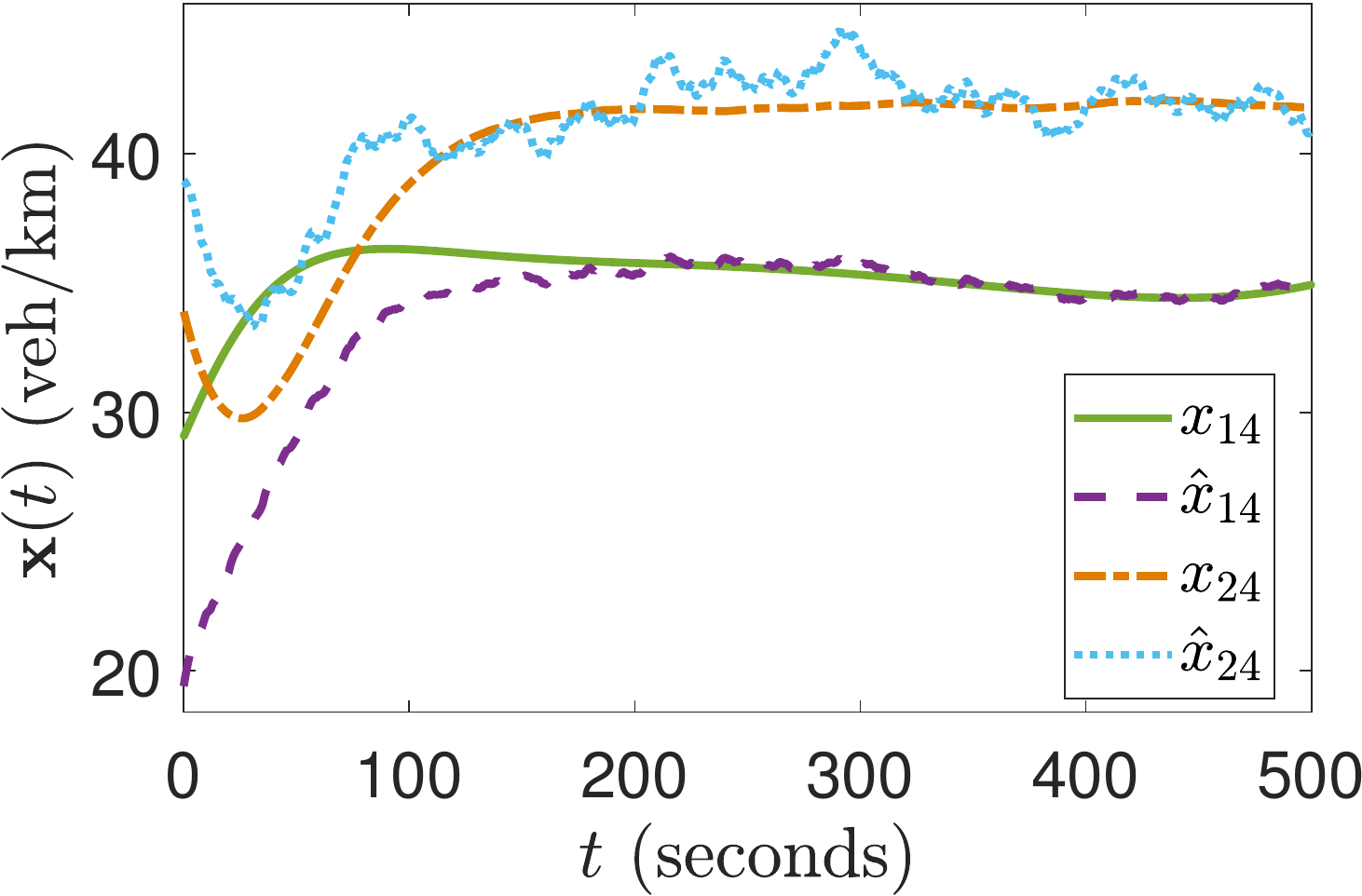}}{}\hspace{-0.0cm}\vspace{-0.1cm}
	\caption{(a) Comparison between the norm of performance output $||{\m z}(t)||_2$ and disturbance where $\zeta = \mu||{\m w}(t)||_{\mathcal{L}_{\infty}}$ for congested (jammed) case and (b) the trajectories of two unmeasured states.}
	\label{fig:congested_uncertainty_1}\vspace{-0.5cm}
\end{figure}
}

\vspace{-0.3cm}
 \color{black}
\subsection{Traffic Density Estimation on A Simple Highway}\label{ssec:test_1}
  The next section showcases the scalability and performance of the proposed methods on a larger system.  Here, we consider a simple highway consisting of and characterized by the following parameters.
  {\color{black}
\begin{itemize}
	\item A total of $n=30$ states, with $N=25$ highway segments, $N_I=3$ on-ramps connected to the $\nth{2}$, $\nth{3}$, and $\nth{4}$ highway segments, and
$N_O=2$ off-ramps connected to the $\nth{22}$ and $\nth{24}$ highway segments.
	\item $p=7$, where $4$ sensors on the $\nth{1}$, $\nth{7}$, $\nth{15}$, and $\nth{25}$ highway segments, one sensor on the $\nth{1}$ on-ramp, and two sensors on both off-ramps such that $\m C \in \mathbb{R}^{7 \times 30}$. It is worth noticing that the considered example is \textit{under-sensed}, in the sense that not all highway segments are equipped with traffic sensors. 
	\item Parameters: $v_f = 31.3$ m/s, $\rho_m = 0.053$ vehicles/m, and $l = 500$ m; the above parameters are adapted from \cite{Contreras2016}, which are obtained from traffic detectors measuring I-15 NB in Las Vegas, Nevada.
	\item For the uncongested (free-flow) case, the exit ratio and traffic flow are chosen such that $\alpha_{1,2} = 0.05$ and $\m u(t) = \bmat{0.2&0.05\m 1_3^{\top}&0.013\m 1_2^{\top}}^{\top}$ for all $t \in [0, 500] \sec$.
	\item For the congested (jammed) case, the exit ratio and traffic flow are chosen such that $\alpha_{1,2} = 0.8$ and $\m u(t) = \bmat{0.25&0.1\m 1_3^{\top}&0.025\m 1_2^{\top}}^{\top}$ for all $t \in [0, 500] \sec$. 
\end{itemize}
}
{\color{black}
Herein we aim to estimate the traffic density on highway segments and on-ramps that are not equipped with sensors. The corresponding Lipschitz constant for the uncongested case is $\gamma_u = 0.5134$, whereas for the congested case is $\gamma_c = 1.0776$, which is obtained from using Eqs. \eqref{eq:lipschitz_const_uncongested} and \eqref{eq:lipschitz_const_congested} given in Propositions \ref{proposition1} and \ref{proposition2}. For the uncongested case, the initial conditions for the system and observer are carefully chosen such that $0 \leq x_i(0) \leq \rho_c$ and $0 \leq \hat{x}_i(0) \leq \rho_c$ for all $i \in {\mathbfcal{E}}$, $i \in \hat{\mathbfcal{E}}$, and $i \in \check{\mathbfcal{E}}$. For the congested case, for all $i \in {\mathbfcal{E}}$, $i \in \hat{\mathbfcal{E}}$, and $i \in \check{\mathbfcal{E}}$ we have $\rho_c < x_i(0) \leq \rho_m$ and $\rho_c < \hat{x}_i(0) \leq \rho_m$. Note that the observer initial conditions are always different than the initial conditions of the system.

{\color{black}
\begin{figure}[t]
\centering 
\subfloat[]{\includegraphics[keepaspectratio=true,scale=0.4]{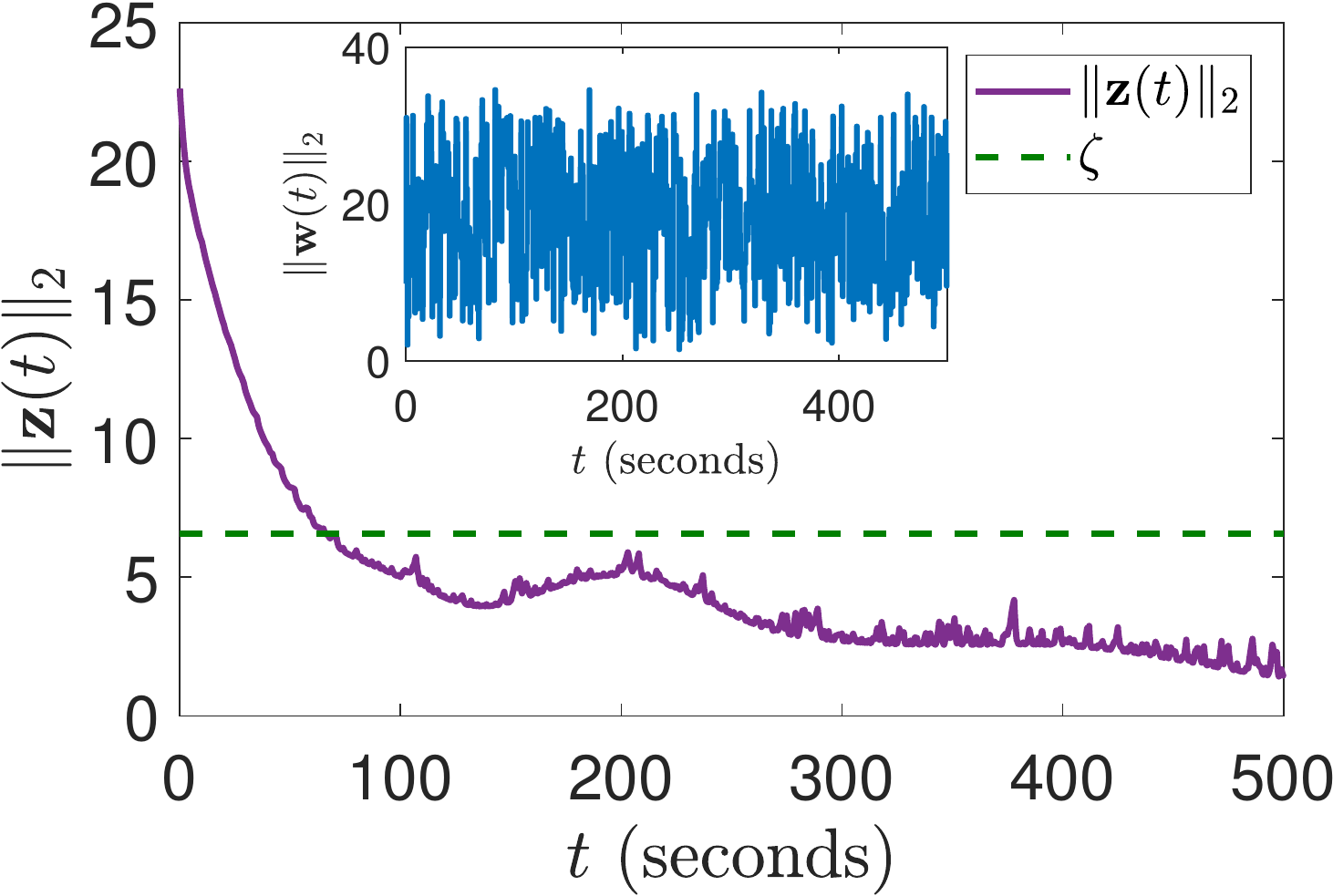}}{}\vspace{-0.3cm}
\subfloat[]{\includegraphics[keepaspectratio=true,scale=0.4]{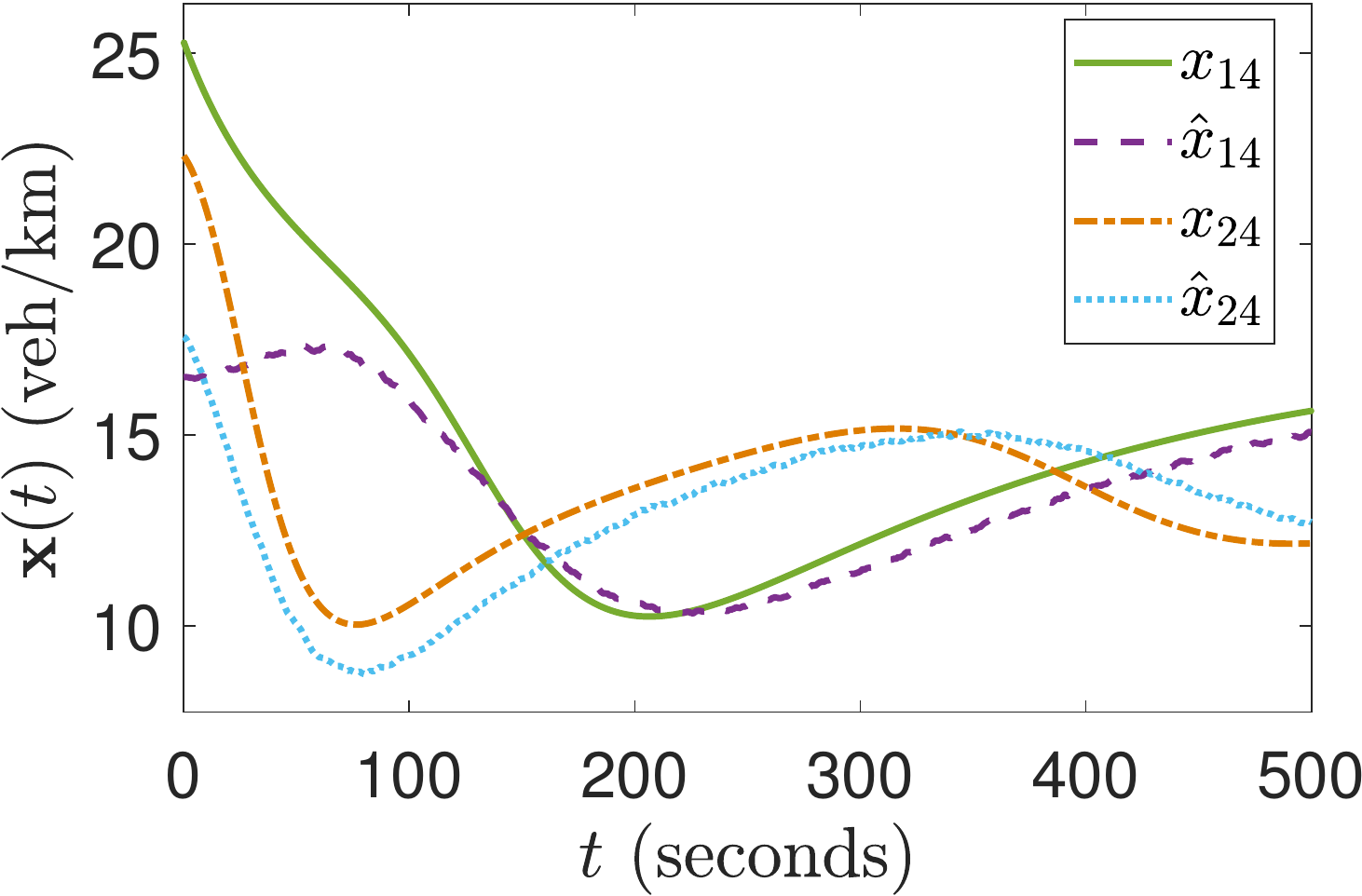}}{}\hspace{-0.0cm}\vspace{-0.1cm}
	\caption{(a) Comparison between the norm of performance output $||{\m z}(t)||_2$ and disturbance where $\zeta = \mu||{\m w}(t)||_{\mathcal{L}_{\infty}}$ for uncongested case with $20\%$ model uncertainty and (b) the trajectories of two unmeasured states.}
	\label{fig:uncongested_uncertainty}\vspace{-0.4cm}
\end{figure}

\begin{figure}[t]
\centering 
\subfloat[]{\includegraphics[keepaspectratio=true,scale=0.4]{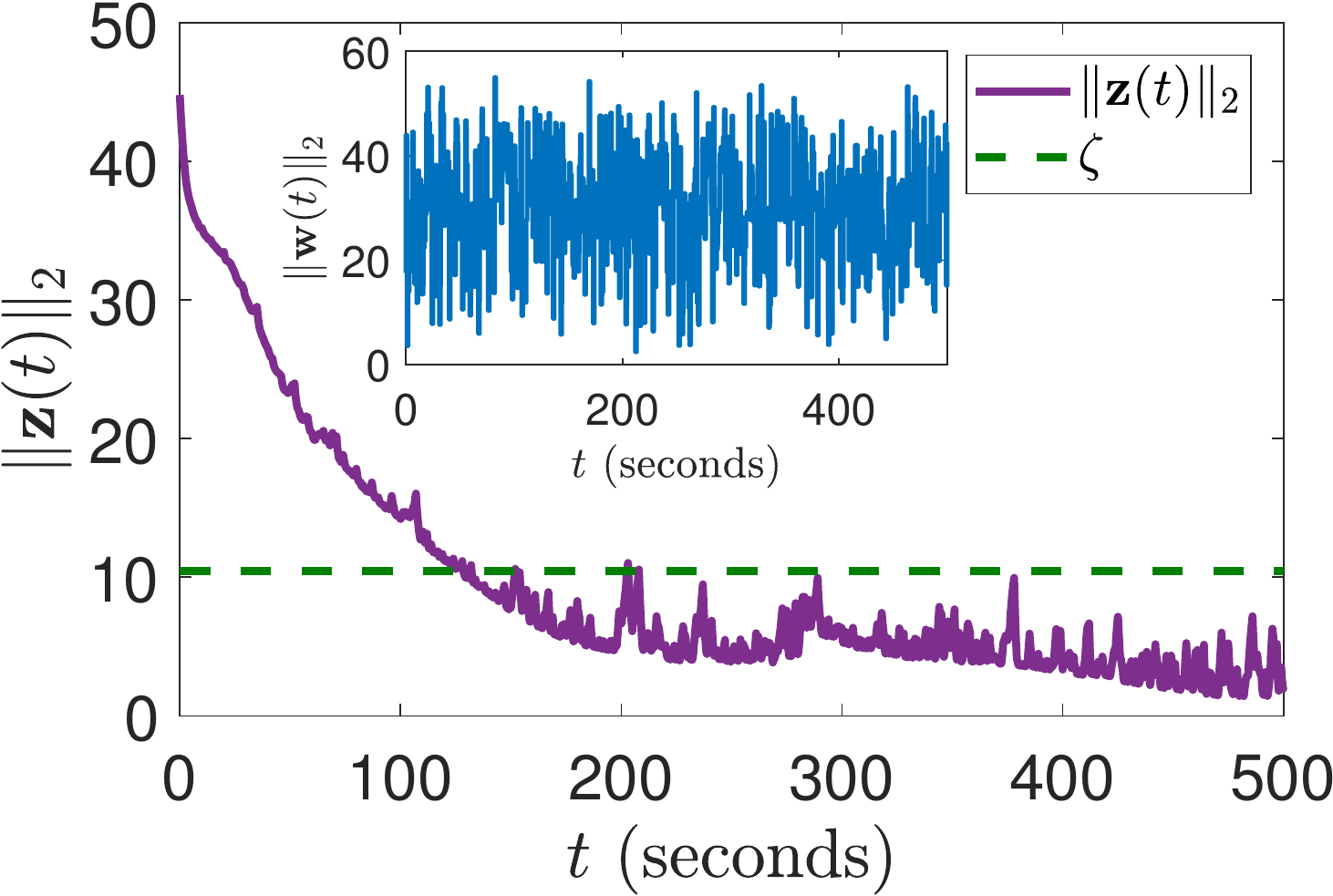}}{}\vspace{-0.3cm}
\subfloat[]{\includegraphics[keepaspectratio=true,scale=0.4]{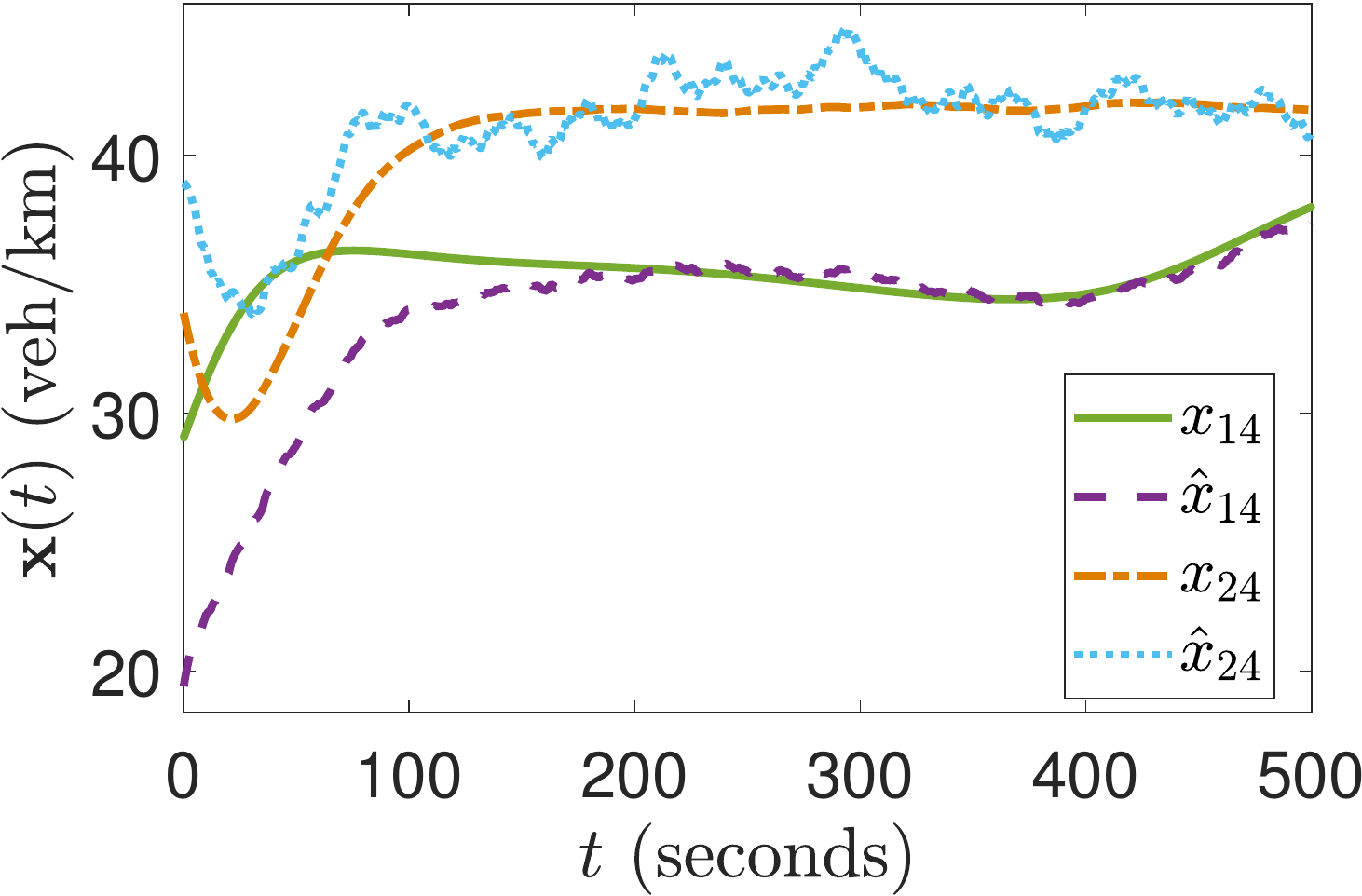}}{}\hspace{-0.0cm}\vspace{-0.1cm}
	\caption{(a) Comparison between the norm of performance output $||{\m z}(t)||_2$ and disturbance where $\zeta = \mu||{\m w}(t)||_{\mathcal{L}_{\infty}}$ for congested case with $20\%$ model uncertainty and (b) the trajectories of two unmeasured states.}
	\label{fig:congested_uncertainty}\vspace{-0.5cm}
\end{figure} 
}

In this simulation, we use the $\mathcal{L}_{\infty}\text{-}\mathrm{Observer}$ described in Theorem \ref{l_inf_theorem}. To obtain a convex problem we set $\alpha = 10^{-3}$ and $\mu_1 = 10^4$. We use SDPNAL+ \cite{Yang2015} to solve the convex problem. The performance matrix is chosen to be $\mZ = \mI$ and specifically, we consider the case when there exist some kind of random disturbances acting as unknown input and measurement noise. The disturbance vector $\m w$ has the following structure
\begin{align}
\m w(t) = 10^3\bmat{0.15\m u(t) \cdot r(t) \\ 0.15\m x(t) \cdot r(t) }.\label{eq:dist_signal_1}
\end{align} 
In \eqref{eq:dist_signal_1}, we define $r:[0,\infty)\rightarrow [-1,1]$ as a random mapping. The corresponding disturbance matrices are chosen so that $\m {B_{\mathrm{w}}} = \bmat{\m {B_{\mathrm{u}}} & \mO}$ and $\m {D_{\mathrm{w}}}= \bmat{\mO & \m {C}}$. This particular choice makes the random parts of unknown input and measurement noise to have maximum values $15\%$ of $\m u(t)$ and $\m x(t)$. Note that the artificially added unknown input $\m w$ is significant and of high magnitude in comparison with the magnitude of the states (the $10^3$ multiplier is due to the change in the units from vehicles/m to vehicles/km). We have considered that to merely test the robustness of the developed estimation method.

The simulations are performed from $t = 0$ to $t_f = 500 \,\mathrm{sec}$ and the results for both cases can be seen in Fig. \ref{fig:uncongested_uncertainty_1} and Fig. \ref{fig:congested_uncertainty_1}. From these figures, we can infer that the observer's trajectories are successfully following the actual system's trajectories. The computed performance index is $\mu = 0.1899$ with the $\mathcal{L}_{\infty}$ norm of the disturbance is $||{\m w}||_{\mathcal{L}_{\infty}}=34.66$ for the uncongested case. For the congested case, $\mu = 0.1899$ and $||{\m w}||_{\mathcal{L}_{\infty}}=55.19$. We observe that the norm of performance outputs ${\m z}$ (where $\m z = \mZ \m e$) for both cases converge, albeit fluctuating due to disturbances, to a close vicinity of zero. Realize that, from these results, the definition of $\mathcal{L}_{\infty}$ stability is actually satisfied, as the norm of the performance output at the steady state region is below $||{\m w}||_{\mathcal{L}_{\infty}}$ multiplied by $\mu$---see Fig. \ref{fig:uncongested_uncertainty_1} (a) and Fig. \ref{fig:congested_uncertainty_1} (a). This corroborates the analytical results given in Theorem~\ref{l_inf_theorem} and showcases that even under high-magnitude disturbances, and using only a sparse combination of traffic sensors, excellent real-time estimates can be generated to approximate the traffic state of a stretched highway. In addition to these results, we actually have also developed a discrete-time version of the $\mathcal{L}_{\infty}$ observer---numerical test results for the CTM with triangular fundamental diagram demonstrates the effectiveness of this approach for traffic density estimation. Nonetheless, we refrain from showing the results here as this paper deals with traffic density estimation on continuous-time traffic model.}

\begin{figure}[t]
	\centering 
	\subfloat[]{\includegraphics[keepaspectratio=true,scale=0.4]{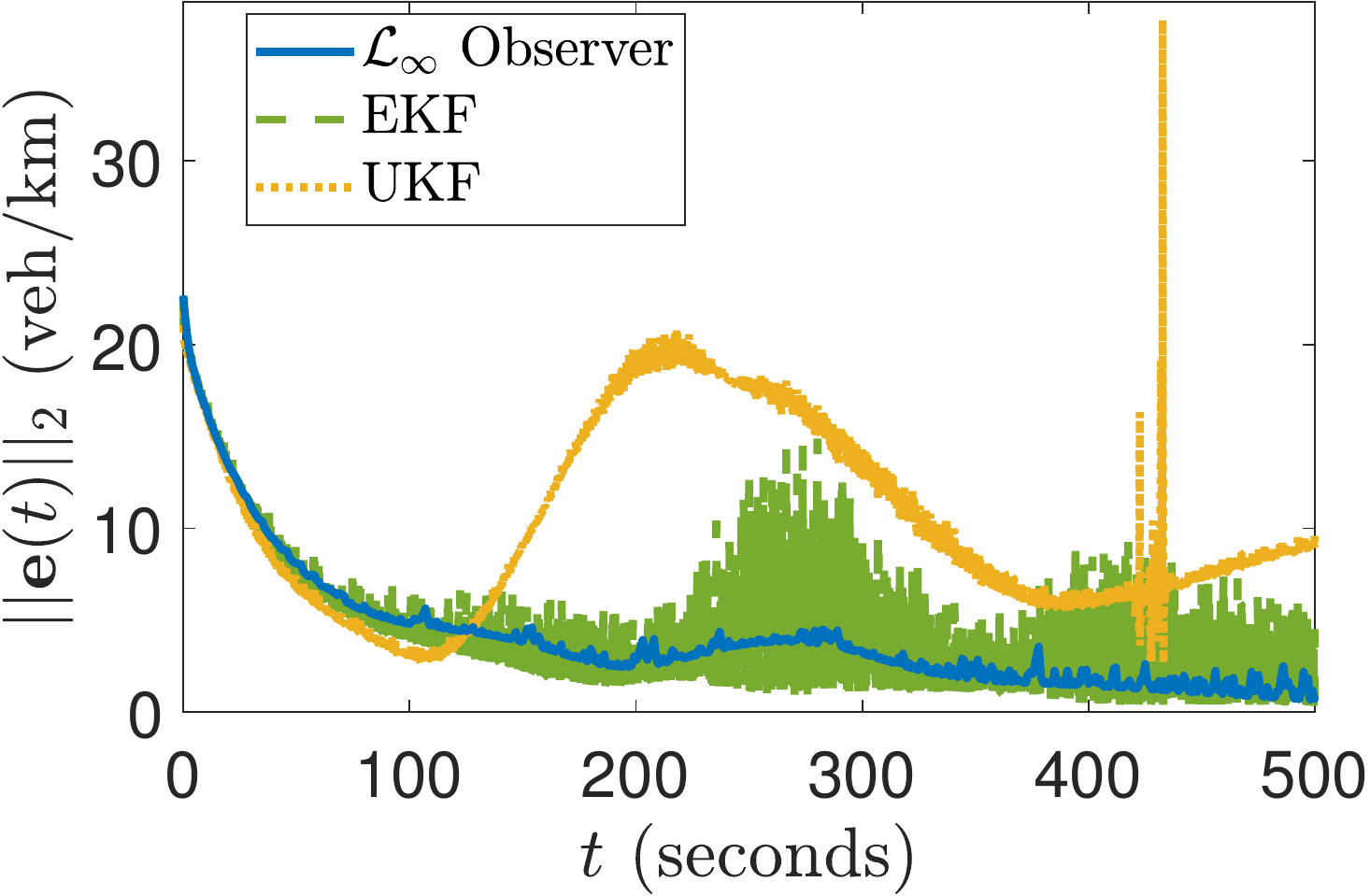}}{}\vspace{-0.3cm}
	\subfloat[]{\includegraphics[keepaspectratio=true,scale=0.4]{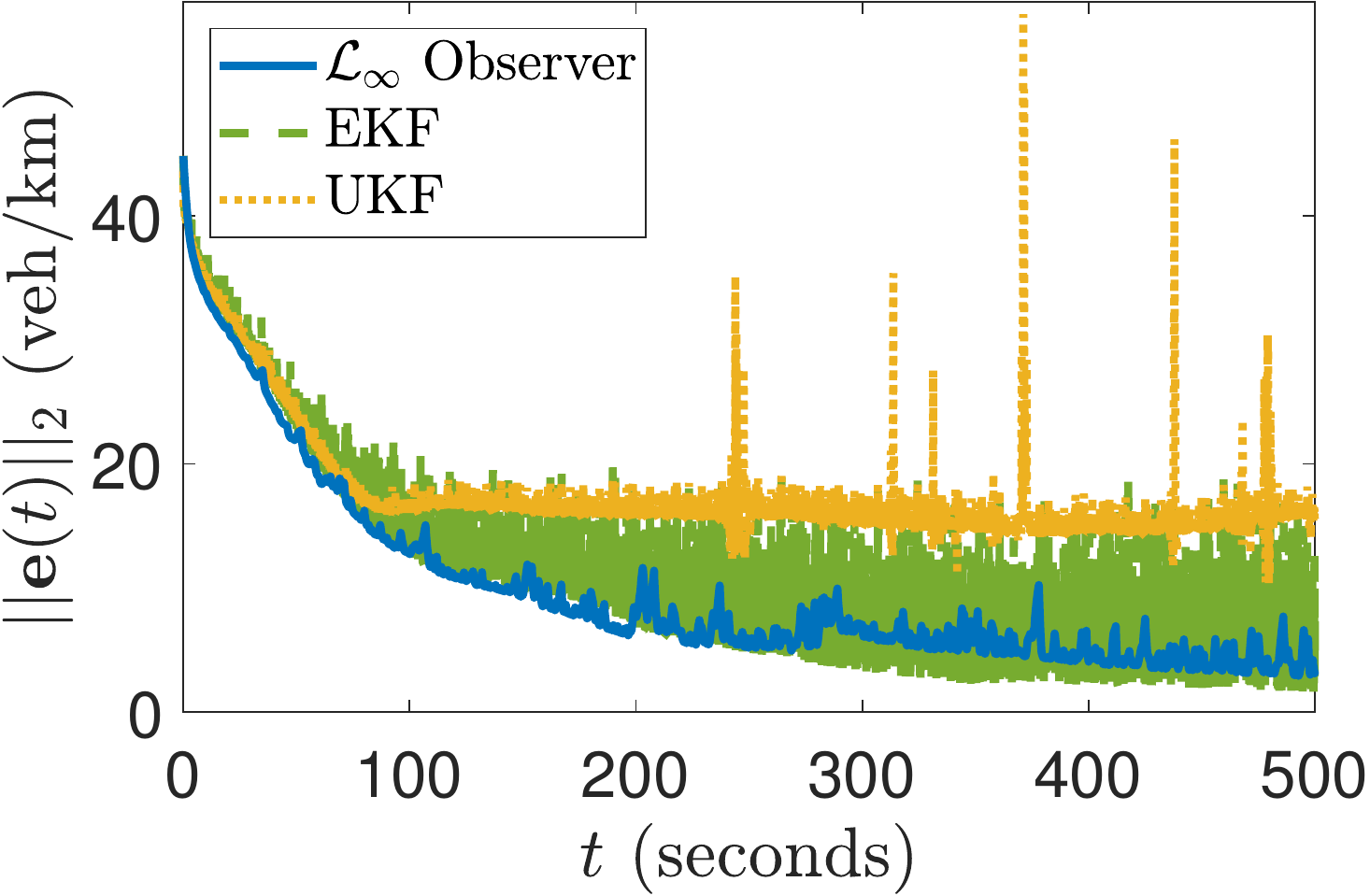}}{}\hspace{-0.0cm}\vspace{-0.1cm}
	\caption{Comparison of estimation error norm between $\mathcal{L}_{\infty}$ observer, EKF, and UKF for Highway \textbf{A} in (a) uncongested case and (b) congested case.}
	\label{fig:EKF_UKF_comparison_30}\vspace{-0.4cm}
\end{figure}

{\color{black}
\subsection{Traffic Density Estimation Under Model Uncertainty}~\label{sec:numerical-param}
We consider model uncertainty and process and measurement noise such that the traffic dynamics can be written as
\begin{subequations}\label{eq:state_space_plant_uncertainty_2}
	\begin{align}
	\dot{\m x}(t) &= (\mA +\Delta\mA)\m x(t) + 	(\mI + \Delta\mI)\m f(\m x)\nonumber \\ & \quad+ (\m {B_{\mathrm{u}}} + \Delta\m {B_{\mathrm{u}}})\m u (t)+\m{B_{\mathrm{w}}}\m {w} (t) \label{eq:state_space_general_a_2}\\
	\m y(t) &= \mC \m x(t)+ \m{D_{\mathrm{w}}}\m {w} (t), \label{eq:state_space_general_b_2}
	\end{align}
\end{subequations}
 where $\m{B_{\mathrm{w}}}$, $\m {D_{\mathrm{w}}}$, and $\m w(t)$ in \eqref{eq:state_space_plant_uncertainty_2} are exactly the same as the ones used in the previous simulation.
 To simulate $20\%$ model parametric additive uncertainty, we use $\kappa = 0.2$ such that $\Delta\mA = \kappa\mA$, $\Delta\mI = \kappa\mI$, and $\Delta\m {B_{\mathrm{u}}}=\kappa \m {B_{\mathrm{u}}}$.
 The $\mathcal{L}_{\infty}$ observer itself has the structure described in \eqref{eq:nonlinear_observer_dynamics}, i.e., without adding the parametric uncertainty.

In this test we again consider the two cases: uncongested and congested. We employ the same highway setup, including the configuration of traffic sensors, as the one described in Section \ref{ssec:test_1}. The results of this numerical test are given in Fig. \ref{fig:uncongested_uncertainty} and Fig. \ref{fig:congested_uncertainty}. From these figures we see that (i) we successfully obtain converging estimation error and (ii) the definition of $\mathcal{L}_{\infty}$ stability is indeed empirically satisfied on both scenarios. The performance indexes for both cases are equal with value $0.01899$. 
The norm of worst case disturbance is $34.64$ for the uncongested case and $55.01$ for the congested case.
From these results, we can conclude that the proposed observer can handle model uncertainty as well as process and measurement noise with acceptable performance for the given performance indexes.  
}

\vspace{-0.545cm}

{\color{black}
\subsection{Comparative Study With Kalman Filter-Based Estimators}
This section is devoted for assessing the performance of the proposed $\mathcal{L}_{\infty}$ observer relative to the performance of other model-driven methods for traffic density estimation. In particular, Kalman filter-based approach for nonlinear systems, such as the Extended Kalman Filter (EKF) and  the Unscented Kalman Filter (UKF), have been extensively utilized with various traffic models and fundamental diagrams for performing traffic state estimation---see \cite[Section 5.1]{seo2017traffic} for a comprehensive survey and discussion. Here we compare our $\mathcal{L}_{\infty}$ observer with EKF and UKF for estimating traffic density on two highway systems of different size. In brief, EKF is a variant of Kalman filter designed specifically for nonlinear systems. EKF has many similarities with Kalman filter except that it utilizes the first-order Taylor approximation to obtain Jacobian matrices of the nonlinear process and measurement models that represent the linearized dynamics around the previous estimated state \cite{ribeiro2004kalman}. Unlike EKF, UKF is a derivative-free state estimator that only relies on the nonlinear process and measurement models of the system and uses an \textit{unscented transformation} to extract, and later to estimate, the mean and covariance data that have gone through a nonlinear transformation---see \cite{wan2000unscented} for a detailed theoretical framework and algorithm of UKF.

The first highway considered in this numerical test, referred to as Highway \textbf{A}, consists of $30$ segments with parameters described in Section \ref{ssec:test_1}. The other highway, named Highway \textbf{B}, consists of a smaller number of segments, which is further detailed as follows.
\begin{itemize}
	\item There are $n=7$ states, with $N=5$ highway segments, $N_I=1$ on-ramp, and $N_O=1$ off-ramp connected to the $\nth{2}$ and $\nth{4}$ highway segments respectively.
	\item $p=2$, where $2$ sensors on the $\nth{1}$ and $\nth{5}$ highway segments.
	\item Parameters such as $v_f$, $\rho_m$, and $l$ are the same as those on Highway \textbf{A}.
	\item The exit ratio and traffic flow for the uncongested case are chosen to be $\alpha_{1} = 0.2$ and $\m u(t) = \bmat{0.1&0.05&0.011}^{\top}$ for all $t \in [0, 500] \sec$.
	\item The exit ratio and traffic flow for the congested case are chosen to be $\alpha_{1} = 0.15$ and $\m u(t) = \bmat{0.34&0.13&0.05}^{\top}$ for all $t \in [0, 500] \sec$. 
\end{itemize}
Since EKF and UKF are based on discrete-time model, we discretize \eqref{eq:state_space_uncongested_ramps} using first-order Taylor approximation with sampling time $T = 0.1$ seconds.
The process and measurement noise covariance matrices for both EKF and UKF are chosen as $\m Q = 10^{-8}\mI$ and $\m R = 10^{-8}\mI$, where the initial error covariance is set to be $\m P_{cov,0} = 10^{-6}\mI$. For UKF, the constants to determine sigma points, which is pivotal in unscented transformation \cite{wan2000unscented}, are set to be $\alpha = 0.1$, $\beta = 2$, and $\kappa = -4$. 

We first note that EKF and UKF both fail to produce converging estimation error when model uncertainty is introduced.\footnote{\color{black}We tested EKF and UKF for a variety of parametric uncertainty magnitudes $\kappa$ akin to Section~\ref{sec:numerical-param}. Unfortunately, both estimators failed to converge and hence our choice to only compare with the $\mathcal{L}_{\infty}$ observer without parametric uncertainty.} As a result, the comparative simulation results shown here are only given under process and measurement noise $\m w(t)$ defined in \eqref{eq:dist_signal_1}. 
The plots for the estimation error norm $\norm
{\m e(t)}_2$ for the three estimation methods are given in Fig. \ref{fig:EKF_UKF_comparison_30} for Highway \textbf{A} and Fig. \ref{fig:EKF_UKF_comparison_7} for Highway \textbf{B} (for the congested and uncongested modes). These figures illustrate that the proposed $\mathcal{L}_{\infty}$ observer has converging error norm with the smallest fluctuation and oscillations, whereas the error norms for EKF and UKF are experiencing much bigger fluctuations. To objectively assess the quality of the estimation and the computational effort in generating state estimates,  Tab.~\ref{tab:comparison_result} produces the computational time and the \textit{Root Mean Square Error} (RMSE) and the \textit{Mean Error} (ME) defined as 
\begin{align*}
	\mathrm{RMSE} &= \sum_{i=1}^{n} \sqrt{\frac{1}{t_f}\sum_{t=1}^{t_f}(e_i(t))^2} \\
		\mathrm{ME} &=\mathrm{average}(\norm{\m e(t)}_2), \; \forall t = t_f-100, t_f-99, \ldots, t_f.
\end{align*}
The ME essentially quantifies the quality of the estimation in the final 100 time-steps of the simulation.
We observe that, other than the dramatically faster simulation time for $\mathcal{L}_{\infty}$ observer relative to those of EKF and UKF, $\mathcal{L}_{\infty}$ observer also produces the smallest \textit{Root Mean Square Error} (RMSE) in most cases.This is due to the fact that the observer dynamics are essentially a one-step predictor with low computational complexity, and the design of the gain $\m L$ is computed offline. 
Moreover, at the end of the simulation, i.e. at $t = 500$ seconds, $\mathcal{L}_{\infty}$ observer also returns the smallest error norm. These results qualitatively and quantitatively show the merits of the proposed $\mathcal{L}_{\infty}$ observer over EKF and UKF. }

\begin{figure}[t]
	\vspace{-0.15cm}
	\centering 
	\subfloat[]{\includegraphics[keepaspectratio=true,scale=0.4]{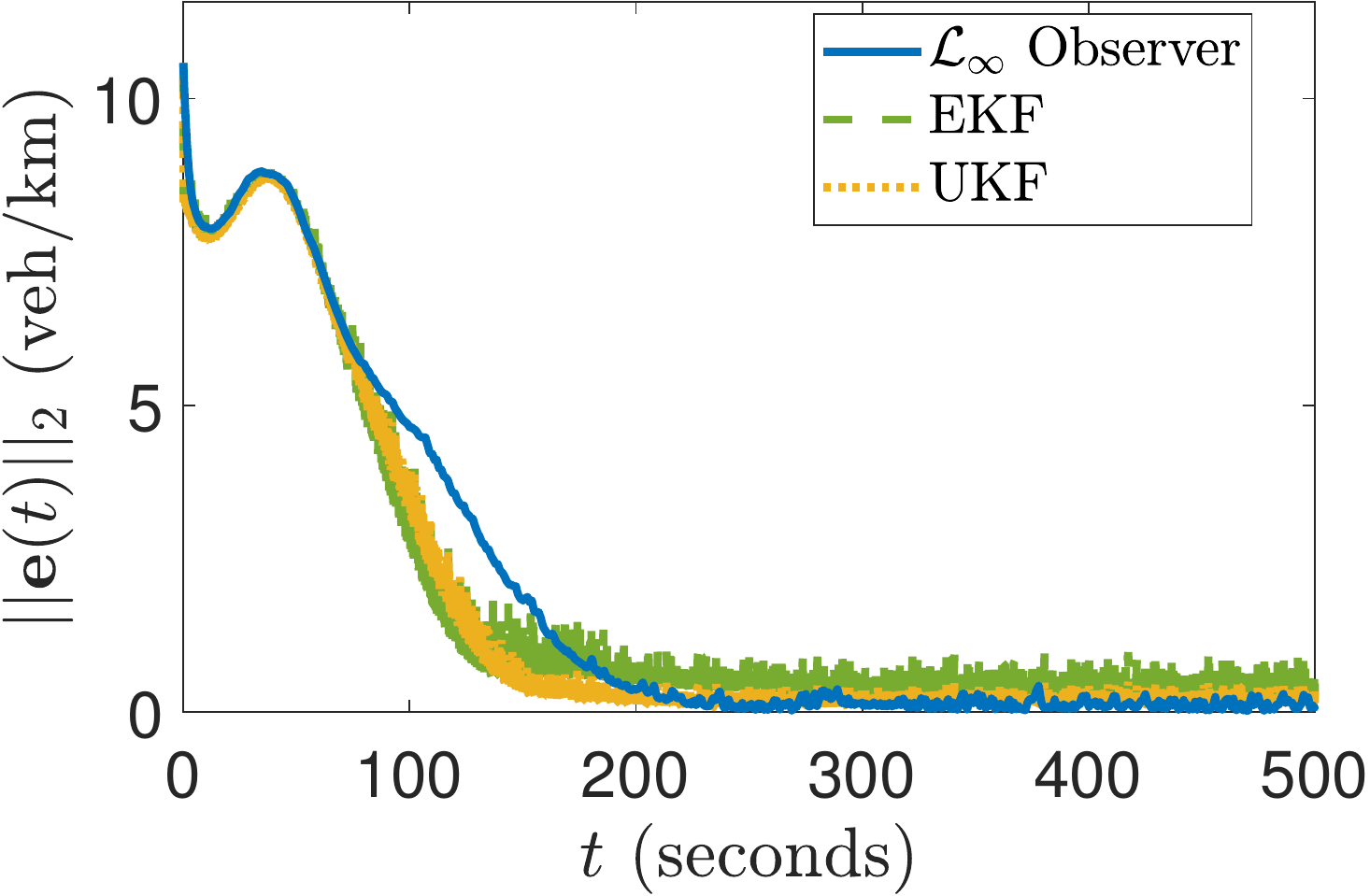}}{}\vspace{-0.3cm}
	\subfloat[]{\includegraphics[keepaspectratio=true,scale=0.4]{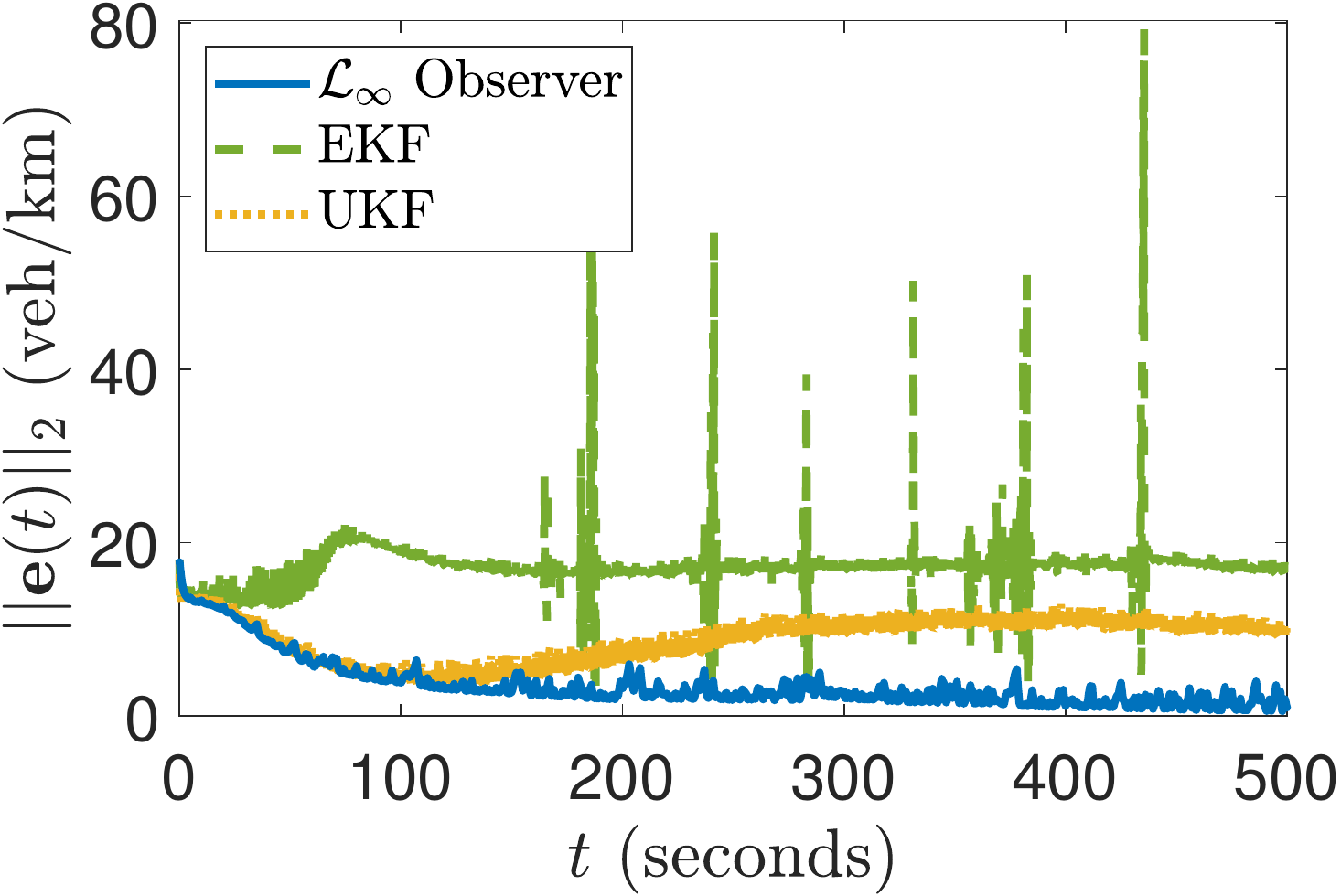}}{}\hspace{-0.0cm}\vspace{-0.12cm}
	\caption{Comparison of estimation error norm between $\mathcal{L}_{\infty}$ observer, EKF, and UKF for Highway \textbf{B} in (a) uncongested case and (b) congested case.}
	\label{fig:EKF_UKF_comparison_7}\vspace{-0.4cm}
\end{figure}

	\begin{table*}[t]
		\footnotesize	\renewcommand{\arraystretch}{1.3}
		\caption{\color{black} Quantitative comparison results between $\mathcal{L}_{\infty}$ observer, EKF, and UKF for both highway systems for the two cases: $\Delta t$ denotes the total computational time; $\mathrm{RMSE}$ quantifies the total error; $\mathrm{ME}$ denotes the mean error norm for $t\in[400,500]\sec$.  }
		\label{tab:comparison_result}
		\centering
		\color{black}	\begin{tabular}{|l|l|l|l|l|l|l|l|l|l|l|l|l|}
			\hline
			\multirow{2}{*}{} & \multicolumn{3}{c|}{Highway \textbf{A}, uncongested} & \multicolumn{3}{c|}{Highway \textbf{A}, congested} & \multicolumn{3}{c|}{Highway \textbf{B}, uncongested} & \multicolumn{3}{c|}{Highway \textbf{B}, congested} \\ \cline{2-13} 
			\textbf{State Estimator}	&   $\Delta t$ (s)    &   $\mathrm{RMSE}$    &  $\mathrm{ME}$     &       $\Delta t$ (s)    &   $\mathrm{RMSE}$    &  $\mathrm{ME}$        &       $\Delta t$ (s)    &   $\mathrm{RMSE}$    &  $\mathrm{ME}$        &       $\Delta t$ (s)    &   $\mathrm{RMSE}$    &  $\mathrm{ME}$     \\ \hline\hline
			$\mathcal{L}_{\infty}$ observer	&   $30.0$   &  $23.72$     &  $1.41$     &  $28.2$     &  $63.03$     &  $4.28$     &  $2.9$     &  $6.88$     &   $0.13$    &   $2.4$    &   $11.35$    &  $1.60$     \\ \hline
			EKF	&   $229.4$    &  $26.84$     &   $2.95$    &  $222.1$     &   $74.68$    &  $7.08$     &     $76.5$     &  $6.35$     &   $0.38$    &   $77.0$    &   $31.47$    &  $17.79$        \\ \hline
			UKF	&  $269.4$     &  $40.37$     &  $7.54$     &  $260.3$     &   $85.49$    &   $15.67$    &     $77.7$      &  $5.93$     &   $0.21$    &   $78.7$    &   $19.60$    &  $10.50$       \\ \hline
		\end{tabular}
	\vspace{-0.2cm}
\end{table*}

\normalcolor

\vspace{-0.25cm}
\subsection{Lipschitz Constant Conservatism and Scalability}
 In this section, and to test for scalability and applicability of the proposed methods to larger systems, we compare the performance of two different solvers, Mosek \cite{andersen2000mosek} and SDPNAL+ \cite{Yang2015}. 
 The comparison is performed by solving the $\mathcal{L}_{\infty}\text{-}\mathrm{Observer}$ from Theorem~\ref{l_inf_theorem} for different highway sizes assuming the uncongested case, ranging from 20 segments of stretched highway to 1000. To obtain a convex problem, we again choose to set the values of $\alpha$ and $\mu_1$ a priori such that $\alpha = 10^{-3}$ and $\mu_1 = 10^4$. The performance matrix for $\mathcal{L}_{\infty}\text{-}\mathrm{Observer}$ is chosen to be $\mZ = 10^{-3}\mI$ with disturbance matrices selected to be $\m {B_{\mathrm{w}}} = \bmat{0.01\m {B_{\mathrm{u}}} & \mO}$ and $\m {D_{\mathrm{w}}}= \bmat{\mO & 0.01\m {C}}$. For simplicity, we impose that all highways only have one on-ramp and one off-ramp, which are respectively connected to the $\nth{2}$ and $(N-1)^{\textrm{th}}$ segments on the stretched highway. All highway segments are assumed to be equipped with sensors except three segments in the middle. Other parameters are similar to those assumed in the previous simulation.

The detailed results of this experiment are given in Tab.~\ref{tab:scalability}. The Lipschitz constant $\gamma_u$ increases as the number of highway segments increases. This is in accordance with Theorem~\ref{proposition1} and Eq. \eqref{eq:lipschitz_const_uncongested}, where the Lipschitz constant $\gamma_u$ is determined by $N$, $N_I$, $N_O$, and $\alpha(\cdot)$. 
We also observe that SDPNAL+ outperforms Mosek in terms of the computational time. SDPNAL+ is able to compute the solution for any highway size within seconds. This is in contrast to Mosek, where the utilized computer can only give results for up to $N = 140$. Additionally, we increase the number of highway segments up to $1000$ for SDPNAL+, and the result is given in Fig. \ref{fig:sdpnal_mosek}, where it only takes less than $200$ seconds for SDPNAL+ to obtain the solution on a personal computer. SDPNAL+ is able to solve the problem with very short computational time as it exploits the sparsity of the traffic model. In fact, SDPNAL+ is designed to deal with SDPs with sparse problem data.  This result is significant as it \textit{(a)} challenges the notion that SDPs do not scale well for dynamic networks of medium to large sizes and \textit{(b)} showcases that the proposed estimator design can be performed whenever a change in the highway mode classification occurs.

\begin{figure}[t]
	\vspace{-0.0cm}	\hspace{0.0cm}\centerline{\includegraphics[keepaspectratio=true,scale=0.4]{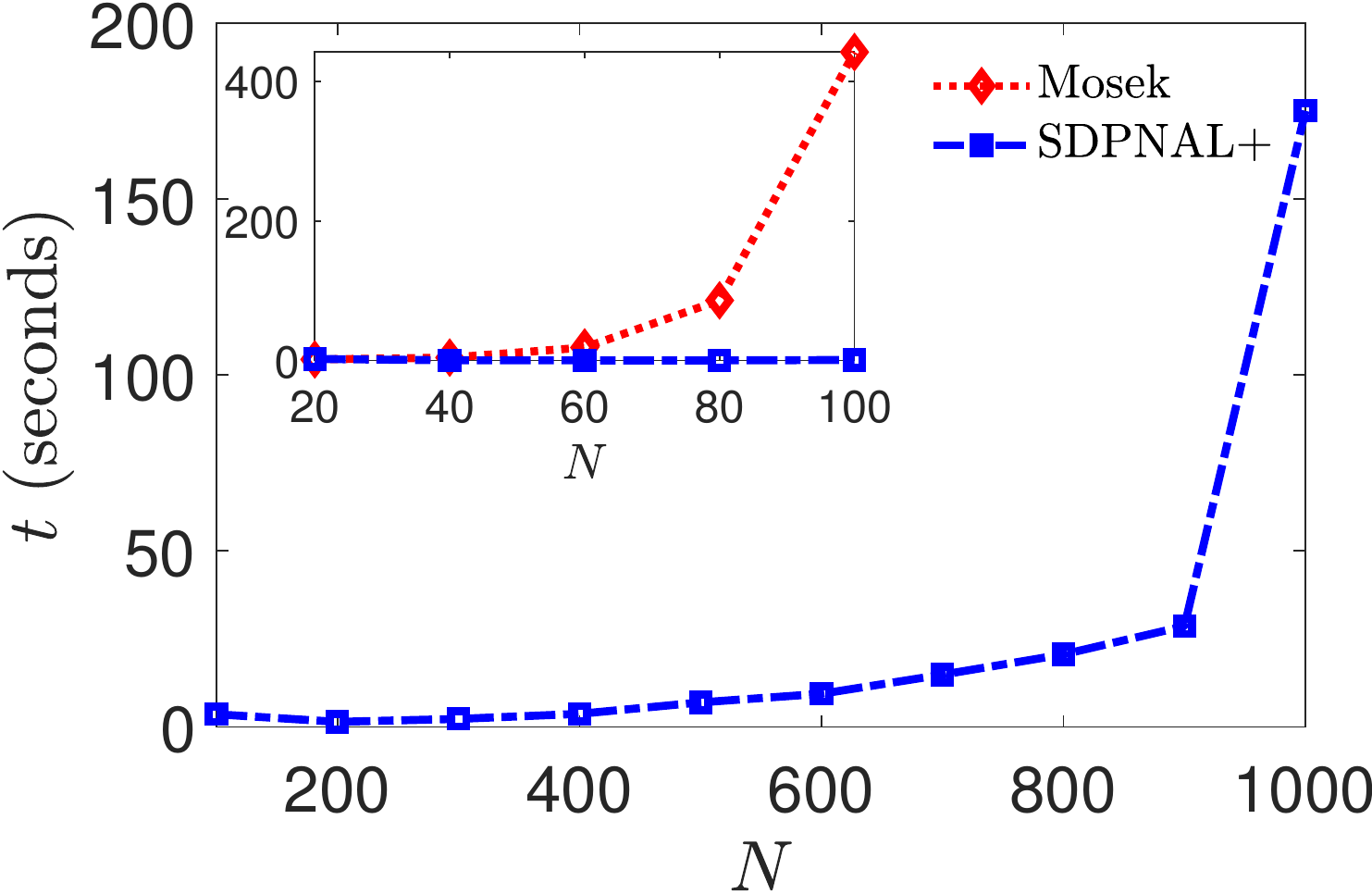}}
		\vspace{-0.5cm}
	\noindent \caption{Computational time of using SDPNAL+ for various highway sizes. The results illustrate that even for a large highway section, the proposed robust estimator can still be used and designed every few minutes.  }
	\label{fig:sdpnal_mosek}
	\vspace{-0.4cm}
\end{figure}

\begin{table}[t!]
	\vspace{-0.2cm}
	\footnotesize	\renewcommand{\arraystretch}{1.3}
	\caption{Lipschitz constant and computational time (for two SDP solvers) for various highway sizes. }
	\label{tab:scalability}
	\centering
	\begin{tabular}{|l|l|l|l|}
		\hline
		${N}$ & ${\gamma_u}$ & {SDPNAL+}  (s)  & {Mosek} (s)\\
		\hline
		\hline
		$20$ & $0.4023$ & $2.2912$ & $1.5930$\\
		\hline
		$40$ & $0.5645$ & $0.3465$ & $4.2073$ \\
		\hline
		$60$ & $0.6895$ & $0.2261$ & $18.6730$ \\
		\hline
		$80$ & $0.7951$ & $0.2147$ & $86.1433$ \\
		\hline
		$100$ & $0.8882$ & $0.8252$ & $441.9469$ \\
		\hline
		$120$ & $0.9724$ & $0.3632$ & $611.3683$ \\
		\hline
		$140$ & $1.0499$  & $0.4390$ & $2279.9862$ \\
		\hline
		$160$ & $1.1221$ & $0.5585$ & $-$ \\
		\hline
		$180$ & $1.1899$ & $0.7072$ & $-$ \\
		\hline
		$200$ & $1.2540$ & $0.8649$ & $-$  \\
		\hline
	\end{tabular}
\vspace{-0.2cm}
\end{table}
\normalcolor

\section{{\color{black}Summary, Limitations, and Future Work}}

{\color{black}
In this paper, we present nonlinear, state-space models for a generalized traffic flow model for stretched highways with arbitrary number of ramp flows based on the \textit{Lighthill Whitham Richards} (LWR) flow model assuming that the stretched highway segments are all either congested or uncongested. We show that the nonlinearities of both traffic models satisfy the locally Lipschitz property and propose analytical methods for computing the corresponding Lipschitz constants. These results are then used to design a real-time and robust observer using the concept of $\mathcal{L}_{\infty}$ stability for stretched highway systems, given a limited number of traffic sensors, under the presence of disturbances, which may include model uncertainty, process noise, and measurement noise. The robust observer, in short, yields an upper bound on the estimation error norm relative to the magnitude to the uncertainty. This upper bound is optimized via convex optimization techniques.

\color{black}
Despite of the above contributions, it is worthwhile to mention that there are several aspects that may potentially restrict the applicability of the proposed framework, such as 
\begin{itemize}
	\item the choice of employing Greenshield's model to represent the fundamental diagram. It is known that in Greenshield's fundamental diagram, the relation between traffic flow and traffic density is somewhat different than the empirical data.
	\item considering a simple stretched highways with ramps instead of the more realistic network of highways or arterial roads. Moreover, it is also assumed that the traffic modes are known, which is another limitation of this work.
	\item using time-invariant state-space equations to model the evolution of traffic density. Realize that in most real situations, some of the parameters on the model are time-varying, such as free-flow speed, maximum density, critical density, exit ratio, and traffic flows.
\end{itemize}
We plan to address the above limitations in our future work by \textit{(i)} building on discrete-time models for traffic density via CTM using a more realistic fundamental diagram and then designing a robust observer for these models, \textit{(ii)} considering the traffic state estimation problem for networks of arterial roads, \textit{(iii)} taking into account the time-varying nature of traffic parameters including the switching behavior. 
We also plan to investigate solutions to the joint problems of \textit{(P1)} placement or selection of static (loop detectors) and dynamic (in-stream sensors from moving vehicles such as GPS data) sensors; and \textit{(P2)} robust state estimation. Problems \textit{(P1-P2)} are coupled, and the outcome of solving them jointly would yield a minimal sensor placement that yields desirable estimation error bounds. 
Finally, the problem of traffic state estimation based on mixed \textit{Eulerian} sensors, such as loop detectors, and \textit{Lagrangian} sensors, such as {GPS} data from moving vehicles, are also of interest for future work---considering that the observability analysis for such problem has been recently studied \cite{contreras2018quality}.\normalcolor
} 


\section*{Acknowledgments}

We gratefully acknowledge the constructive comments from the editor and the reviewers. We also acknowledge the financial support from the National Science Foundation through Grants 1636154, 1728629, 1917164, 1917056 and Valero Energy Corporation through Valero PhD Competitive Research Scholarship Awards.

\bibliographystyle{IEEEtran}	\bibliography{IEEEabrv,bibl}

\onecolumn

\appendix

\subsection{State-Space Parameters of Traffic Density Dynamic Model for the Congested Case}\label{apdx:table_congested}
The state-space parameters of traffic density model for the uncongested and congested cases are given in Tabs.~\ref{tab:dynamics_parameter_uncongested} and \ref{tab:dynamics_parameter_congested}.

\begin{table}[h]
	\footnotesize \centering 
	\caption{Detailed parameters for traffic density dynamic model for the uncongested case}
	\label{tab:dynamics_parameter_uncongested}
	\renewcommand{\arraystretch}{1.5}
	\begin{tabular}{|l|l|}
		\hline 
		\textbf{ Parameter} & \textbf{Description } \\[0.5ex] 
		\hline 
		\hline
		$\begin{array}{c}
		\hspace{-0.0cm}\m A_1 \in \mathbb{R}^{N\times N},\; \\ \m A_2 \in \mathbb{R}^{N\times (N_I+N_O)}, \; \\
		\m A_3 \in \mathbb{R}^{(N_I+N_O)\times(N_I+N_O)}
		\end{array}$ & $\begin{array} {ccc}
		\m A_1  = \bmat{-\frac{v_f}{l}&0&0&\cdots&0\\
			\frac{v_f}{l}&-\frac{v_f}{l}&0&\cdots&0\\
			0&\frac{v_f}{l}&-\frac{v_f}{l}&\cdots&0\\
			\vdots&\vdots&\vdots&\ddots&\vdots\\
			0&0&0&\cdots&-\frac{v_f}{l}} \label{eq:uncongested_dynamic_matrix_1-a}  \end{array} \begin{array} {lcr}
		\m A_2(i,j) = \begin{cases}
		\vphantom{\left(\frac{v_f}{l}\right)}\frac{v_f}{l}, & \text{if } i\in\mathbfcal{E}_I,\;j\in\hat{\mathbfcal{E}}\\
		-\frac{\alpha(\bar{j})v_f}{l}, & \text{if } i\in\mathbfcal{E}_O,\;j = N_I + \bar{j},\;\bar{j}\in\check{\mathbfcal{E}}\\
		0,  & \text{otherwise}
		\end{cases} \label{eq:uncongested_dynamic_matrix_1-b} 
		\\
		\\
		\m A_3(i,j) = \begin{cases}
		\vphantom{\left(\frac{v_f}{l}\right)}-\frac{v_f}{l}, & \text{if } i = j, i\in\hat{\mathbfcal{E}}\\
		\frac{\alpha(\bar{i})v_f}{l}, & \text{if } i = j, \;i = N_I + \bar{i}, \bar{i}\in\check{\mathbfcal{E}}\\
		0,  & \text{otherwise}.
		\end{cases} 
		\end{array}$\\ 
		\hline		
		$\begin{array}{l}
		\m f:\mathbb{R}^n\rightarrow\mathbb{R}^n 
		\end{array}$ & $\begin{array} {lcr}
		f_i(\m x) = \begin{cases}
		\vphantom{\left(\frac{v_f}{l}\right)}\delta  x_i^2\label{eq:uncongested_dynamic_matrix_1-d}, & \text{if } i\in\mathbfcal{E}\setminus {\mathbfcal{E}}_I \cup {\mathbfcal{E}}_O,\;i = 1\\
		\delta  \left( x_i^2-x_{i-1}^2\right) \label{eq:uncongested_dynamic_matrix_1-e}, & \text{if } i\in\mathbfcal{E}\setminus {\mathbfcal{E}}_I \cup {\mathbfcal{E}}_O$, $i \neq 1\\
		\delta  \left( x_i^2-x_{i-1}^2-x_j^2\right)\label{eq:uncongested_dynamic_matrix_1-f}, & \text{if }  i\in\mathbfcal{E}_I\setminus{\mathbfcal{E}}_I \cap {\mathbfcal{E}}_O,\;j = N+\bar{j},\; \bar{j}\in\hat{\mathbfcal{E}} \\
		\delta  \left( x_i^2-x_{i-1}^2+\alpha(\bar{j})x_j^2\right) \label{eq:uncongested_dynamic_matrix_1-g} , & \text{if } i\in\mathbfcal{E}_O\setminus{\mathbfcal{E}}_I \cap {\mathbfcal{E}}_O,\;j = N+N_I+\bar{j},\; \bar{j}\in\check{\mathbfcal{E}} \\
		\delta  \left( x_i^2-x_{i-1}^2-x_j^2+\alpha(\bar{k})x_k^2\right) \label{eq:uncongested_dynamic_matrix_1-h}, & \text{if } i\in{\mathbfcal{E}}_I \cap {\mathbfcal{E}}_O,\;j = N+\bar{j},\; \bar{j}\in\hat{\mathbfcal{E}},\;k = N+N_I+\bar{k},\; \bar{k}\in\check{\mathbfcal{E}} \\
		\delta  x_i^2 \label{eq:uncongested_dynamic_matrix_1-i}, & \text{if }i = N + \bar{i},\;\bar{i}\in\hat{\mathbfcal{E}}\\
		-\alpha(\bar{i})\delta  x_i^2 \label{eq:uncongested_dynamic_matrix_1-j}	, & \text{if }	i = N + N_I + \bar{i},\;\bar{i}\in\check{\mathbfcal{E}}
		\end{cases}  
		\end{array}$ \\ 
		\hline		
		$\begin{array}{c}
		\m {B_{\mathrm{u}}} \in \mathbb{R}^{n\times(1+N_I+N_O)}
		\end{array}$ & $\begin{array} {ccc}
		\m {B_{\mathrm{u}}}(i,j) = \begin{cases}
		\frac{1}{l}, & \text{if } i = j = 1, \;i\in{\mathbfcal{E}}\\
		\frac{1}{l}, & \text{if } i = N+\bar{i},\;j = 1+\bar{i},\;\bar{i}\in\hat{\mathbfcal{E}}\\
		-\frac{1}{l}, & \text{if } i = N+N_I+\bar{i},\;j = 1+N_I+\bar{i},\; \bar{i} \in\check{\mathbfcal{E}}\\
		0,  & \text{otherwise} 
		\end{cases}   \label{eq:uncongested_dynamic_matrix_1-k}
		\end{array}$\\ 
		\hline	
		$\begin{array}{c}
		\m u \in \mathbb{R}^{1+N_I+N_O}
		\end{array}$ & $\begin{array} {ccc}
		\m u(t)= \bmat{f_{\mathrm{in}} & \hat{f}_1 & \hat{f}_2 & \cdots & \hat{f}_{N_I} & \check{f}_1 & \check{f}_2 & \cdots & \check{f}_{N_O}}^{\top} \label{eq:uncongested_dynamic_matrix_1-l}
		\end{array}$\\ 
		\hline	
	\end{tabular}
\end{table}

\begin{table}[h]
	\footnotesize \centering 
	\caption{Detailed parameters for traffic density dynamic model for the congested case}
	\label{tab:dynamics_parameter_congested}
	\renewcommand{\arraystretch}{1.5}
	\begin{tabular}{|l|l|}
		\hline 
		\textbf{ Parameter} & \textbf{Description } \\[0.5ex] 
		\hline 
		\hline
		$\begin{array}{c}
		\hspace{-0.0cm}\m A_1 \in \mathbb{R}^{N\times N}
		\end{array}$ & $\begin{array} {ccc}
		\m A_1 &= \bmat{\frac{v_f}{l}&-\frac{v_f}{l}&0&\cdots&0\\
			0&\frac{v_f}{l}&-\frac{v_f}{l}&\cdots&0\\
			0&0&\frac{v_f}{l}&\cdots&0\\
			\vdots&\vdots&\vdots&\ddots&\vdots\\
			0&0&0&\cdots&\frac{v_f}{l}} \label{eq:congested_dynamic_2a} \end{array}$\\ 
		\hline	
		$\begin{array}{l}
		\m f:\mathbb{R}^n\rightarrow\mathbb{R}^n 
		\end{array}$ & $\begin{array} {lcr}
		f_i(\m x) = \begin{cases}
		\vphantom{\left(\frac{v_f}{l}\right)}-\delta  x_i^2\label{eq:congested_dynamic_2b}, & \text{if } i\in\mathbfcal{E}\setminus {\mathbfcal{E}}_I \cup {\mathbfcal{E}}_O,\;i = N\\
		\delta  \left( x_{i+1}^2-x_{i}^2\right)\label{eq:congested_dynamic_2c}, & \text{if } i\in\mathbfcal{E}\setminus {\mathbfcal{E}}_I \cup {\mathbfcal{E}}_O$, $i \neq N\\
		\delta  \left( x_{i+1}^2-x_{i}^2-x_j^2\right)\label{eq:congested_dynamic_2d}, & \text{if }  i\in\mathbfcal{E}_I\setminus{\mathbfcal{E}}_I \cap {\mathbfcal{E}}_O,\;j = N+\bar{j},\; \bar{j}\in\hat{\mathbfcal{E}} \\
		\delta  \left( x_{i+1}^2-x_{i}^2+\alpha(\bar{j})x_j^2\right)\label{eq:congested_dynamic_2e} , & \text{if } i\in\mathbfcal{E}_O\setminus{\mathbfcal{E}}_I \cap {\mathbfcal{E}}_O,\;j = N+N_I+\bar{j},\; \bar{j}\in\check{\mathbfcal{E}} \\
		\delta  \left( x_{i+1}^2-x_{i}^2-x_j^2+\alpha(\bar{k})x_k^2\right)\label{eq:congested_dynamic_2f}, & \text{if } i\in{\mathbfcal{E}}_I \cap {\mathbfcal{E}}_O,\;j = N+\bar{j},\; \bar{j}\in\hat{\mathbfcal{E}},\;k = N+N_I+\bar{k},\; \bar{k}\in\check{\mathbfcal{E}} \\
		\delta  x_i^2 \label{eq:uncongested_dynamic_matrix_1-i}, & \text{if }i = N + \bar{i},\;\bar{i}\in\hat{\mathbfcal{E}}\\
		-\alpha(\bar{i})\delta  x_i^2 \label{eq:uncongested_dynamic_matrix_1-j}	, & \text{if }	i = N + N_I + \bar{i},\;\bar{i}\in\check{\mathbfcal{E}}
		\end{cases}  
		\end{array}$ \\ 
		\hline		
		$\begin{array}{c}
		\m {B_{\mathrm{u}}} \in \mathbb{R}^{n\times(1+N_I+N_O)}
		\end{array}$ & $\begin{array} {ccc}
		\m {B_{\mathrm{u}}}(i,j) = \begin{cases}
		-\frac{1}{l}, & \text{if } i = N,\;j = 1,\;i\in{\mathbfcal{E}}\\
		\frac{1}{l}, & \text{if } i = N+\bar{i},\;j = 1+\bar{i},\;\bar{i}\in\hat{\mathbfcal{E}}\\
		-\frac{1}{l}, & \text{if } i = N+N_I+\bar{i},\;j = 1+N_I+\bar{i},\;\bar{i} \in\check{\mathbfcal{E}}\\
		0,  & \text{otherwise}.
		\end{cases}   \label{eq:congested_dynamic_2i}
		\end{array}$\\ 
		\hline	
		$\begin{array}{c}
		\m u \in \mathbb{R}^{1+N_I+N_O}
		\end{array}$ & $\begin{array} {ccc}
		\m u(t)= \bmat{f_{\mathrm{out}} & \hat{f}_1 & \hat{f}_2 & \cdots & \hat{f}_{N_I} & \check{f}_1 & \check{f}_2 & \cdots & \check{f}_{N_O}}^{\top} \label{eq:uncongested_dynamic_matrix_2j}
		\end{array}$\\ 
		\hline	
	\end{tabular}
	\vspace{-0.2cm}
\end{table}
\twocolumn

\subsection{Proof of Proposition \ref{proposition1}}\label{apdx:proposition1_proof}
{\color{black}
\begin{proof}
	Let $\m f:\mathbb{R}^n\rightarrow\mathbb{R}^n$ be such function. By using the fact that $v_f,\rho_m,\rho_c,l > 0$ and $\rho_c = \tfrac{1}{2}\rho_m$, then for each case of $f_i(\cdot)$ specified in Tab.~\ref{tab:dynamics_parameter_uncongested} and any $\m x, \hat{\m x} \in \mathbfcal{X}_{\m{\mathrm{u}}}$ we have
	\begin{subequations}
		\begin{enumerate}[label=$\alph*$)]
			\item $i\in\mathbfcal{E}\setminus {\mathbfcal{E}}_I \cup {\mathbfcal{E}}_O,\;i = 1$
			\begin{align*}
			\abs{ f^a_i(\m x)-f^a_i(\hat{\m x}) }
			&\leq \frac{v_f}{l}\abs{ x_i-\hat{x}_i}.
			\end{align*} 
			Since
			$\abs{ x_i-\hat{x}_i}^2 \leq \sum_{j = 1}^{n}\abs{ x_j-\hat{x}_j}^2 = \norm{\m x-\m \hat{\m x}}_2^2$, then
			\begin{align}
			\abs{ f^a_i(\m x)-f^a_i(\hat{\m x}) } 
			&\leq \frac{v_f}{l}\norm{\m x-\m \hat{\m x}}_2. \label{eq:thrm1_1a}
			\end{align} 
			
			\item $i\in\mathbfcal{E}\setminus {\mathbfcal{E}}_I \cup {\mathbfcal{E}}_O$, $i \neq 1$ 
			\begin{align*}
			\abs{ f^b_i(\m x)-f^b_i(\hat{\m x}) } 
			&\leq \frac{v_f}{l}\left(\abs{ x_i-\hat{x}_i}+\,\abs{ x_{i-1}-\hat{x}_{i-1}}\right).
			\end{align*} 
			Since 
			$\sum^{k = i}_{k = i-1}\abs{x_{k}-\hat{x}_{k}}\leq \sqrt{2} \norm{\m x - \hat{\m x} }_2$
			for any $1 < i \leq N$, then 
			\begin{align}
			\abs{ f^b_i(\m x)-f^b_i(\hat{\m x}) } 
			&\leq \frac{\sqrt{2} v_f}{l}\norm{\m x-\m \hat{\m x}}_2.\label{eq:thrm1_1b}
			\end{align}
			
			\item $i\in\mathbfcal{E}_I\setminus{\mathbfcal{E}}_I \cap {\mathbfcal{E}}_O,\;j = N+\bar{j},\; \bar{j}\in\hat{\mathbfcal{E}}$
			\begin{align*}
			\abs{ f^c_i(\m x)-f^c_i(\hat{\m x}) } 
			\leq &\frac{v_f}{l}(\abs{ x_i-\hat{x}_i}+\abs{ x_{i-1}-\hat{x}_{i-1}})\\ &+\frac{2v_f}{l}\abs{ x_{j}-\hat{x}_{j}}.
			\end{align*} 
Since $\sum^{k = i}_{k = i-1}\abs{x_{k}-\hat{x}_{k}}\leq \sqrt{2} \norm{\m x - \hat{\m x} }_2$ 
			for any $1 < i \leq N$, then 
			\begin{align}
			\abs{ f^c_i(\m x)-f^c_i(\hat{\m x}) } 
			&\leq \frac{(2+\sqrt{2})v_f}{l}\norm{\m x-\m \hat{\m x}}_2.\label{eq:thrm1_1c}
			\end{align}
			
			\item $i\in\mathbfcal{E}_O\setminus{\mathbfcal{E}}_I \cap {\mathbfcal{E}}_O,\;j = N+N_I+\bar{j},\; \bar{j}\in\check{\mathbfcal{E}}$
			\begin{align*}
			\abs{ f^d_i(\m x)-f^d_i(\hat{\m x}) } 
			\leq &\frac{v_f}{l}(\abs{ x_i-\hat{x}_i}+\abs{ x_{i-1}-\hat{x}_{i-1}}) \\&+\frac{2v_f}{l}\alpha(\bar{j})\abs{ x_{j}-\hat{x}_{j}}.
			\end{align*} 
			Since 
			$\sum^{k = i}_{k = i-1}\abs{x_{k}-\hat{x}_{k}}\leq \sqrt{2} \norm{\m x - \hat{\m x} }_2$
			for any $1 < i \leq N$, then 
			\begin{align}
			\abs{ f^d_i(\m x)-f^d_i(\hat{\m x}) } 
			&\leq \frac{(\sqrt{2}+2\alpha(\bar{j}))v_f}{l}\norm{\m x-\m \hat{\m x}}_2.\label{eq:thrm1_1d}
			\end{align}
			
			\item $i\in{\mathbfcal{E}}_I \cap {\mathbfcal{E}}_O,\;j = N+\bar{j},\; \bar{j}\in\hat{\mathbfcal{E}},\;k = N+N_I+\bar{k},\; \bar{k}\in\check{\mathbfcal{E}}$
			\begin{align*}
			\abs{ f^e_i(\m x)-f^e_i(\hat{\m x}) } 
			\leq &\frac{v_f}{l}(\abs{ x_i-\hat{x}_i}+\abs{ x_{i-1}-\hat{x}_{i-1}}) \\ &+\frac{2v_f}{l}(\abs{ x_{j}-\hat{x}_{j}}+\alpha(\bar{k})\abs{ x_{k}-\hat{x}_{k}}). 
			\end{align*} 
			Since $\sum^{l = i}_{l = i-1}\abs{x_{l}-\hat{x}_{l}}\leq \sqrt{2} \norm{\m x - \hat{\m x} }_2$
			for any $1 < i \leq N$, then 
			\begin{align}
			\hspace{-0.4cm}\abs{ f^e_i(\m x)-f^e_i(\hat{\m x}) } 
			&\leq \frac{(2+\sqrt{2}+2\alpha(\bar{j}))v_f}{l}\norm{\m x-\m \hat{\m x}}_2.\label{eq:thrm1_1e}
			\end{align}
			
			\item $i = N + \bar{i},\;\bar{i}\in\hat{\mathbfcal{E}}$
			\begin{align}
			\hspace{-0.4cm}|{ f^f_i(\m x)-f^f_i(\hat{\m x}) }|
			&\leq \frac{2v_f}{l}\abs{ x_i-\hat{x}_i} \leq \frac{2v_f}{l}\norm{\m x-\m \hat{\m x}}_2.\label{eq:thrm1_1f}
			\end{align}
			
			\item $i = N + N_I + \bar{i},\;\bar{i}\in\check{\mathbfcal{E}}$
			\begin{align}
				\hspace{-0.4cm}\abs{ f^g_i(\m x)-f^g_i(\hat{\m x}) }
			&\leq \frac{2\alpha(\bar{i})v_f}{l}\abs{ x_i-\hat{x}_i} \leq \frac{2\alpha(\bar{i})v_f}{l}\norm{\m x-\m \hat{\m x}}_2.\label{eq:thrm1_1g}
			\end{align}
		\end{enumerate}
	\end{subequations}
	From \cref{eq:thrm1_1a,eq:thrm1_1b,eq:thrm1_1c,eq:thrm1_1d,eq:thrm1_1e,eq:thrm1_1f,eq:thrm1_1g}, we know that for any function $f^z_i(\cdot)$ where $z\in \lbrace a,b,\hdots,g\rbrace$, there exists $\gamma_i\geq 0$ such that $\abs{f_i(\m x)- f_i(\hat{\m x})}  \leq \gamma_i \norm{\m x - \hat{\m x} }_2$. Since it holds that 
	\begin{align}
	\norm{\m f(\m x)-\m f(\hat{\m x})}_2^2
	&= \sum_{i=1}^{n}\abs{f_i(\m x)- f_i(\hat{\m x})}^2 \leq \sum_{i=1}^{n}\gamma_i^2 \norm{\m x - \hat{\m x} }_2^2,\label{eq:lemma}
	\end{align}
	then $\m f(\cdot)$ is locally Lipschitz in $\mathbfcal{X}_{\m{\mathrm{u}}}$ with Lipschitz constant given in \eqref{eq:lipschitz_const_uncongested}.
\end{proof}	
}

\subsection{Proof of Proposition \ref{proposition2}}\label{apdx:proposition2_proof}
{\color{black}
	\begin{proof}
		Let $\m f:\mathbb{R}^n\rightarrow\mathbb{R}^n$ be such function. By using the fact that $v_f,\rho_m,\rho_c,l > 0$ and $\rho_c = \tfrac{1}{2}\rho_m$, then for each case of $f_i(\cdot)$ specified in Tab.~\ref{tab:dynamics_parameter_congested} and any $\m x, \hat{\m x} \in \mathbfcal{X}_{\m{\mathrm{c}}}$ we have
		\begin{subequations}
			\begin{enumerate}[label=$\alph*$)]
				\item $i\in\mathbfcal{E}\setminus {\mathbfcal{E}}_I \cup {\mathbfcal{E}}_O,\;i = N$
				\begin{align*}
				\abs{ f^a_i(\m x)-f^a_i(\hat{\m x}) }
				&\leq \frac{2v_f}{l}\abs{ x_i-\hat{x}_i}.
				\end{align*} 
				Since
				$\abs{ x_i-\hat{x}_i}^2 \leq \sum_{j = 1}^{n}\abs{ x_j-\hat{x}_j}^2 = \norm{\m x-\m \hat{\m x}}_2^2$, then
				\begin{align}
				\abs{ f^a_i(\m x)-f^a_i(\hat{\m x}) } 
				&\leq \frac{2v_f}{l}\norm{\m x-\m \hat{\m x}}_2. \label{eq:thrm1_2a}
				\end{align} 
				
				\item $i\in\mathbfcal{E}\setminus {\mathbfcal{E}}_I \cup {\mathbfcal{E}}_O$, $i \neq N$ 
				\begin{align*}
				\abs{ f^b_i(\m x)-f^b_i(\hat{\m x}) } 
				&\leq \frac{2v_f}{l}\left(\abs{ x_{i+1}-\hat{x}_{i+1}}+\,\abs{ x_{i}-\hat{x}_{i}}\right).
				\end{align*} 
				Since 
				$\sum^{k = i+1}_{k = i}\abs{x_{k}-\hat{x}_{k}}\leq \sqrt{2} \norm{\m x - \hat{\m x} }_2$
				for any $1 \leq i < N$, then 
				\begin{align}
				\abs{ f^b_i(\m x)-f^b_i(\hat{\m x}) } 
				&\leq \frac{2\sqrt{2} v_f}{l}\norm{\m x-\m \hat{\m x}}_2.\label{eq:thrm1_2b}
				\end{align}
				
				\item $i\in\mathbfcal{E}_I\setminus{\mathbfcal{E}}_I \cap {\mathbfcal{E}}_O,\;j = N+\bar{j},\; \bar{j}\in\hat{\mathbfcal{E}}$
				\begin{align*}
				\abs{ f^c_i(\m x)-f^c_i(\hat{\m x}) } 
				\leq &\frac{2v_f}{l}(\abs{ x_{i+1}-\hat{x}_{i+1}}+\,\abs{ x_{i}-\hat{x}_{i}}\\ &+\abs{ x_{j}-\hat{x}_{j}}).
				\end{align*} 
				Since 
				$\sum^{k = i+1}_{k = i}\abs{x_{k}-\hat{x}_{k}} + \abs{x_{j}-\hat{x}_{j}} \leq 2 \norm{\m x - \hat{\m x} }_2$
				for any $1 \leq i < N$, then 
				\begin{align}
				\abs{ f^c_i(\m x)-f^c_i(\hat{\m x}) } 
				&\leq \frac{4v_f}{l}\norm{\m x-\m \hat{\m x}}_2.\label{eq:thrm1_2c}
				\end{align}
				
				\item $i\in\mathbfcal{E}_O\setminus{\mathbfcal{E}}_I \cap {\mathbfcal{E}}_O,\;j = N+N_I+\bar{j},\; \bar{j}\in\check{\mathbfcal{E}}$
				\begin{align*}
				\abs{ f^d_i(\m x)-f^d_i(\hat{\m x}) } 
				\leq &\frac{2v_f}{l}(\abs{ x_i-\hat{x}_i}+\abs{ x_{i-1}-\hat{x}_{i-1}} \\&+ \alpha(\bar{j})\abs{ x_{j}-\hat{x}_{j}}).
				\end{align*} 
				Since 
				\hspace{-0.25cm}
				$\sum^{k = i+1}_{k = i}\abs{x_{k}-\hat{x}_{k}}\leq \sqrt{2} \norm{\m x - \hat{\m x} }_2$
				for any $1 \leq i < N$, then 
				\begin{align}
				\abs{ f^d_i(\m x)-f^d_i(\hat{\m x}) } 
				&\leq \frac{2(\sqrt{2}+\alpha(\bar{j}))v_f}{l}\norm{\m x-\m \hat{\m x}}_2.\label{eq:thrm1_2d}
				\end{align}
				
				\item $i\in{\mathbfcal{E}}_I \cap {\mathbfcal{E}}_O,\;j = N+\bar{j},\; \bar{j}\in\hat{\mathbfcal{E}},\;k = N+N_I+\bar{k},\; \bar{k}\in\check{\mathbfcal{E}}$
				\begin{align*}
				\abs{ f^e_i(\m x)-f^e_i(\hat{\m x}) } 
				\leq &\frac{2v_f}{l}(\abs{ x_i-\hat{x}_i}+\abs{ x_{i-1}-\hat{x}_{i-1}} \\ &+\abs{ x_{j}-\hat{x}_{j}})+\frac{2v_f}{l}\alpha(\bar{k})\abs{ x_{k}-\hat{x}_{k}}. 
				\end{align*} 
				Since $\sum^{l = i+1}_{l = i}\abs{x_{l}-\hat{x}_{l}} + \abs{x_{j}-\hat{x}_{j}} \leq 2 \norm{\m x - \hat{\m x} }_2$
				for any $1 \leq i < N$, then 
				\begin{align}
				\hspace{-0.4cm}\abs{ f^e_i(\m x)-f^e_i(\hat{\m x}) } 
				&\leq \frac{2(2+\alpha(\bar{j}))v_f}{l}\norm{\m x-\m \hat{\m x}}_2.\label{eq:thrm1_2e}
				\end{align}
				
				\item $i = N + \bar{i},\;\bar{i}\in\hat{\mathbfcal{E}}$
				\begin{align}
				\hspace{-0.4cm}|{ f^f_i(\m x)-f^f_i(\hat{\m x}) }|
				&\leq \frac{2v_f}{l}\abs{ x_i-\hat{x}_i} \leq \frac{2v_f}{l}\norm{\m x-\m \hat{\m x}}_2.\label{eq:thrm1_2f}
				\end{align}
				
				\item $i = N + N_I + \bar{i},\;\bar{i}\in\check{\mathbfcal{E}}$
				\begin{align}
				\hspace{-0.4cm}\abs{ f^g_i(\m x)-f^g_i(\hat{\m x}) }
				&\leq \frac{2\alpha(\bar{i})v_f}{l}\abs{ x_i-\hat{x}_i} \leq \frac{2\alpha(\bar{i})v_f}{l}\norm{\m x-\m \hat{\m x}}_2.\label{eq:thrm1_2g}
				\end{align}
			\end{enumerate}
		\end{subequations}
		From \cref{eq:thrm1_2a,eq:thrm1_2b,eq:thrm1_2c,eq:thrm1_2d,eq:thrm1_2e,eq:thrm1_2f,eq:thrm1_2g}, we know that for any function $f^z_i(\cdot)$ where $z\in \lbrace a,b,\hdots,g\rbrace$, there exists $\gamma_i\geq 0$ such that $\abs{f_i(\m x)- f_i(\hat{\m x})}  \leq \gamma_i \norm{\m x - \hat{\m x} }_2$. By using \eqref{eq:lemma}, it follows that
		$\m f(\cdot)$ is locally Lipschitz in $\mathbfcal{X}_{\m{\mathrm{c}}}$ with Lipschitz constant given in \eqref{eq:lipschitz_const_congested}.
	\end{proof}	
}

\subsection{Proof of Theorem \ref{l_inf_theorem}}\label{apdx:thrm2_proof}
{\color{black}
\begin{proof}
	Let $\m z = \m Z\m e$ be the performance output for estimation error $\m e$ and $\m w$ be an unknown bounded disturbance. Construct $V(\m e) = \m e^{\top}\mP \m e$ as the Lyapunov function candidate  where $\mP \succ 0$.  
	From \cite[Theorem 1]{pancake2002analysis}, it can be shown that the estimation error dynamics \eqref{eq:est_error_dynamics} is $\mathcal{L}_{\infty}$ stable with performance level $\mu = \sqrt{\mu_0\mu_1+\mu_2}$ if there exist 
	constants $\mu_0,\mu_1,\mu_2\in \mathbb{R}_{+}$ such that 
	\begin{subequations}\label{eq:l_inf_lemma}
		\vspace{-0.5cm}
		\begin{align}
		\mu_0 \norm{\m w}_2^2 &< V(\m e) \;\Rightarrow \dot{V}(\m e) < 0 \label{eq:l_inf_lemma_1}\\
		\norm{\m z}_2^2 &\leq \mu_1V(\m e)+\mu_2\norm{\m w}_2^2,\label{eq:l_inf_lemma_2}
		\end{align}
	\end{subequations}
		for all $t\geq 0$.
	Note that \eqref{eq:l_inf_lemma_1} holds if there exists $\alpha > 0$ such that $\dot{V}(\m e) +\alpha\left(V(\m e)-\mu_0 \norm{\m w}_2^2\right)\leq 0$. From here, we obtain
	\begin{align}
		\dot{V}(\m e) +\alpha\left(V(\m e)-\mu_0 \norm{\m w}_2^2\right)&\leq 0 \nonumber \\
		\Leftrightarrow \dot{\m e}^{\top}\mP\m e + {\m e}^{\top}\mP\dot{\m e} + \alpha \m e^{\top}\mP\m e -\alpha\mu_0\m w^{\top}\m w &\leq 0, \nonumber \\
		\Leftrightarrow 
		\bmat{\m \Omega &*&*\\
			\mP&\mO&*\\\m {B_{\mathrm{w}}}^{\top}\mP-\m {D_{\mathrm{w}}}^{\top}\mL^{\top}\mP & \mO &-\alpha\mu_0\mI}&\preceq 0, \label{eq:l_inf_theorem_proof_1} 
	\end{align}
	where $\m \Omega \triangleq \mA^{\top}\mP + \mP\mA -\mC^{\top}\mL^{\top}\mP-\mP\mL\mC +\alpha\mP$. Since the function $\m f(\cdot)$ is locally Lipschitz, then we also have
	\begin{align}
		\norm{\Delta\m f}_2^2 &\leq \gamma^2 \norm{\m e}_2^2 \nonumber \\
		\Leftrightarrow \Delta\m f^{\top}\Delta\m f-\gamma^2\m e^{\top} \m e &\leq 0 \nonumber \\
		\Leftrightarrow \bmat{-\gamma^2\mI &*&*\\
			\mO&\mI&*\\\mO&\mO&\mO} &\preceq 0.\label{eq:l_inf_theorem_proof_2}
	\end{align}
	Applying the S-procedure Lemma to \eqref{eq:l_inf_theorem_proof_1} from \eqref{eq:l_inf_theorem_proof_2} for $\epsilon \geq 0$ yields 
	\begin{align*}
		\bmat{\m \Omega + \epsilon\gamma^2\mI &*&*\\
			\mP&-\epsilon\mI&*\\\m {B_{\mathrm{w}}}^{\top}\mP-\m {D_{\mathrm{w}}}^{\top}\mL^{\top}\mP & \mO &-\alpha\mu_0\mI}\preceq 0.
	\end{align*}
	Defining $\mY \triangleq \mP\mL$ and $\m\Psi \triangleq \m \Omega + \epsilon\gamma^2\mI $, the above is equivalent to \eqref{eq:l_inf_theorem_1}. Next, substituting $\m z = \mZ \m e$ to \eqref{eq:l_inf_lemma_2} yields
	\begin{align}
		\norm{\mZ \m e}_2^2 -\mu_1V(\m e)-\mu_2\norm{\m w}_2^2&\leq 0 \nonumber\\
		\Leftrightarrow\m e^{\top}\mZ^{\top}\mZ \m e-\mu_1\m e^{\top}\mP\m e -\mu_2 \m w^{\top} \m w &\leq 0. \nonumber
	\end{align}
	By using congruence transformation and applying  the Schur Complement, the above is equivalent to \eqref{eq:l_inf_theorem_2}. Thus, the solvability of optimization problem \eqref{eq:l_inf_theorem} ensures that \eqref{eq:l_inf_lemma} is satisfied, which consequently implies that the estimation error dynamics given in \eqref{eq:est_error_dynamics} is $\mathcal{L}_{\infty}$ stable with performance level $\mu = \sqrt{\mu_0\mu_1+\mu_2}$ and observer gain $\mL = \mP^{-1}\mY$. 
\end{proof}
}

\begin{IEEEbiography}
	[{\includegraphics[width=1in,height=1.25in,clip,keepaspectratio]{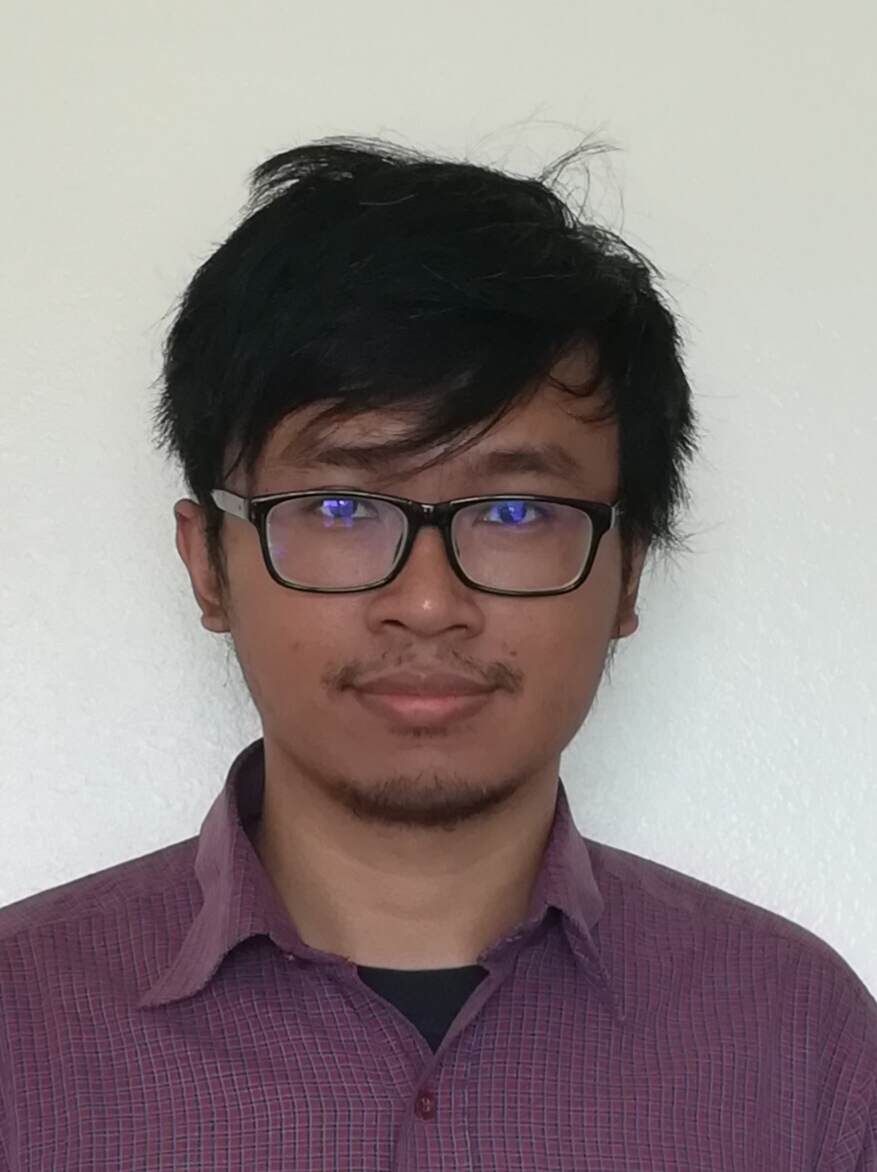}}]
	{Sebastian A. Nugroho} was born in Yogyakarta, Indonesia and received the B.S. and M.S. degrees in Electrical Engineering from Institut Teknologi Bandung (ITB), Indonesia in 2012 and 2014. He is currently a graduate research assistant and pursuing the Ph.D. degree in Electrical Engineering at the University of Texas, San Antonio (UTSA), USA. 
	His main areas of research interest are control theory, state estimation, and engineering optimization with applications to cyber-physical systems. 
	He received the Valero PhD Competitive Research Scholarship Awards from 2017 to 2019. 
\end{IEEEbiography} 

\vskip -2\baselineskip plus -1fil

\begin{IEEEbiography}
	[{\includegraphics[width=1in,height=1.25in,clip,keepaspectratio]{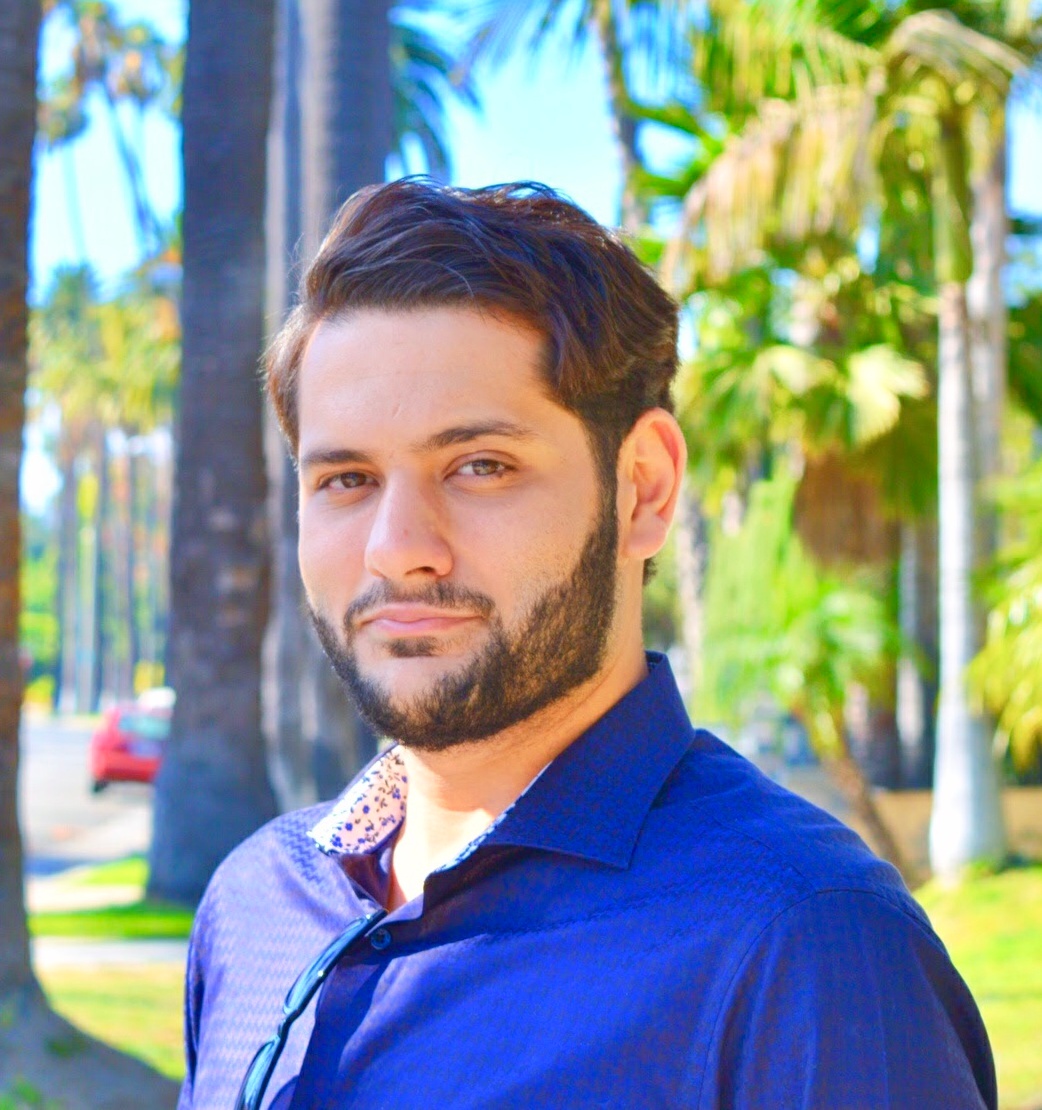}}]
	{Ahmad F. Taha}  is  an assistant professor with the Department of Electrical and Computer Engineering at the University of Texas, San Antonio. He received the B.E. and Ph.D. degrees in Electrical and Computer Engineering from the American University of Beirut, Lebanon in 2011 and Purdue University, West Lafayette, Indiana in 2015. Dr. Taha is interested in understanding how complex cyber-physical systems (CPS) operate, behave, and \textit{misbehave}. His research focus includes optimization, control, and security of CPSs with applications to power, water, and transportation networks. Dr. Taha is an editor of IEEE Transactions on Smart Grid and the editor of the IEEE Control Systems Society Electronic Letter (E-Letter).
\end{IEEEbiography}

\vskip -2\baselineskip plus -1fil

\begin{IEEEbiography}
	[{\includegraphics[width=1in,height=1.25in,clip,keepaspectratio]{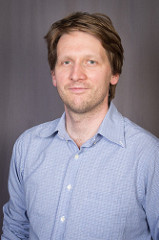}}]
	{Christian Claudel} is an Assistant Professor of Civil, Architectural and Environmental Engineering at UT-Austin. He received the PhD degree in Electrical Engineering from UC-Berkeley in 2010, and the MS degree in Plasma Physics from Ecole Normale Superieure de Lyon in 2004. He received the Leon Chua Award from UC-Berkeley in 2010 for his work on the Mobile Millennium traffic monitoring system. His research interests include control and estimation of distributed parameter systems, wireless sensor networks and unmanned aerial vehicles.
\end{IEEEbiography}

\end{document}